\documentclass[american,aps,prl,reprint,superscriptaddress,showpacs,footinbib]{revtex4-1}
\usepackage[T1]{fontenc}
\usepackage[latin9]{inputenc}
\usepackage{color}
\usepackage{bbm}
\usepackage{babel}
\usepackage{amsthm}
\usepackage{amsmath}
\usepackage{amssymb}
\usepackage{graphics,graphicx}
\usepackage{verbatim}
\usepackage[unicode=true,pdfusetitle,
 bookmarks=true,bookmarksnumbered=false,bookmarksopen=false,
 breaklinks=false,pdfborder={0 0 0},backref=false,colorlinks=true]
 {hyperref}
\hypersetup{
 linkcolor=urlcolor,citecolor=urlcolor,urlcolor=urlcolor}

\usepackage{tikz}
\usepackage{tikz-cd}

\definecolor{urlcolor}{rgb}{0,0,0.7}

\makeatletter

\@ifundefined{textcolor}{}
{%
 \definecolor{BLACK}{gray}{0}
 \definecolor{WHITE}{gray}{1}
 \definecolor{RED}{rgb}{1,0,0}
 \definecolor{GREEN}{rgb}{0,1,0}
 \definecolor{BLUE}{rgb}{0,0,1}
 \definecolor{CYAN}{cmyk}{1,0,0,0}
 \definecolor{MAGENTA}{cmyk}{0,1,0,0}
 \definecolor{YELLOW}{cmyk}{0,0,1,0}
 }
 
 \ifx\proof\undefined
\newenvironment{proof}[1][\protect\proofname]{\par
\normalfont\topsep6\p@\@plus6\p@\relax
\trivlist
\itemindent\parindent
\item[\hskip\labelsep
\scshape
#1]\ignorespaces
}{%
\endtrivlist\@endpefalse
}
\providecommand{\proofname}{Proof}
\fi

\theoremstyle{plain}
\newtheorem{theorem}{\protect\theoremname}
\theoremstyle{plain}
\newtheorem{lemma}[theorem]{\protect\lemmaname}
\theoremstyle{plain}

\theoremstyle{plain}

\theoremstyle{plain}
\newtheorem{corollary}[theorem]{\protect\corollaryname}

\theoremstyle{plain}
\newtheorem{thm}{\protect\theoremname}
\theoremstyle{plain}
\newtheorem{lem}[thm]{\protect\lemmaname}
\theoremstyle{plain}
\newtheorem{cor}[thm]{\protect\corollaryname}
\theoremstyle{plain}

\theoremstyle{plain}

\theoremstyle{definition}
\newtheorem{remark}{\protect\remarkname}
\newtheorem{example}{\protect\examplename}

\newtheorem{thm1}{Theorem}
\newtheorem{lem1}[thm1]{Lemma}

\usepackage{babel}

\usepackage{txfonts}\usepackage{braket}

\providecommand{\theoremname}{Theorem}
\providecommand{\lemmaname}{Lemma}
\providecommand{\propositionname}{Proposition}
\providecommand{\definitionname}{Definition}
\providecommand{\corollaryname}{Corollary}
\providecommand{\remarkname}{Remark}
\providecommand{\examplename}{Example}

\newcommand{\mbb}{\mathbb}
\newcommand{\mc}{\mathcal}

\newcommand{\norm}[1]{\left|\left|#1\right|\right|}
\newcommand{\op}[2]{|#1\rangle \langle #2|}
\newcommand{\ip}[2]{\langle #1|#2\rangle}
\newcommand{\proj}[1]{| #1 \rangle\!\langle #1 |}
\newcommand{\tr}{\mathrm{tr}}
\newcommand{\pr}{\mathrm{Pr}}
\newcommand{\ovl}[1]{\overline{#1}}

\begin{document}

\title{Relating the Resource Theories of Entanglement and Quantum Coherence}

\author{Eric Chitambar}

\affiliation{Department of Physics and Astronomy, Southern Illinois University,
Carbondale, Illinois 62901, USA}

\author{Min-Hsiu Hsieh}

\affiliation{Centre for Quantum Computation \& Intelligent Systems (QCIS), Faculty of Engineering and Information Technology (FEIT), University of Technology Sydney (UTS), NSW 2007, Australia}

\begin{abstract}
Quantum coherence and quantum entanglement represent two fundamental features of non-classical systems that can each be characterized within an operational resource theory.  In this paper, we unify the resource theories of entanglement and coherence by studying their combined behavior in the operational setting of local incoherent operations and classical communication (LIOCC).  Specifically we analyze the coherence and entanglement trade-offs in the tasks of state formation and resource distillation.  For pure states we identify the minimum coherence-entanglement resources needed to generate a given state, and we introduce a new LIOCC monotone that completely characterizes a state's optimal rate of bipartite coherence distillation.  This result allows us to precisely quantify the difference in operational powers between global incoherent operations, LIOCC, and local incoherent operations \textit{without} classical communication.   Finally, a bipartite mixed state is shown to have distillable entanglement if and only if entanglement can be distilled by LIOCC, and we strengthen the well-known Horodecki criterion for distillability.  

\end{abstract}

\maketitle

The ability for quantum systems to exist in ``superposition states'' reveals the wave-like nature of matter and represents a strong departure from classical physics.  Systems in such superposition states are often said to possess quantum coherence.  
There has currently been much interest in constructing a resource theory of quantum coherence \cite{Aberg-2006a, Levi-2014a, Baumgratz-2014a, Bromley-2015a, Korzekwa-2015a, Yuan-2015a, Winter-2015a, Singh-2015a, Marvian-2015a, Streltsov-2015c, Yadin-2015a}, in part because of recent experimental and numerical findings that suggest quantum coherence alone can enhance or impact physical dynamics in biology \cite{Lloyd-2011a, Li-2012a, Huelga-2013a, Lambert-2013a}, transport theory \cite{Rebentrost-2009a, Bjorn-2013a, Levi-2014a}, and thermodynamics \cite{Lostaglio-2015a, Narasimhachar-2015a}.  

In a standard resource-theoretic treatment of quantum coherence, the free (or ``incoherent'') states are those that are diagonal in some fixed reference (or ``incoherent'') basis 
Different classes of allowed (or ``incoherent'') operations have been proposed in the literature \cite{Aberg-2006a, Baumgratz-2014a, Marvian-2015a, Streltsov-2015c, Yadin-2015a} (see also \cite{Chitambar-2016a, Marvian-2016a} for comparative studies of these approaches), however an essential requirement is that the incoherent operations act invariantly on the set of diagonal density matrices.  Incoherent operations can then be seen as one of the most basic generalizations of classical operations (i.e.~stochastic maps) since their action on diagonal states can always be simulated by classical processing.  Note also that most experimental setups will have a natural basis to work in, and arbitrary unitary time evolutions might be physically difficult to implement.  In these settings, there are practical advantages to identifying ``diagonal preserving'' operations as being ``free'' relative to coherent-generating ones. 

One does not need to look far to find an important connection between incoherent operations and quantum entanglement, the latter being one of the most important resources in quantum information processing \cite{Horodecki-2009a}.  Consider the task of entanglement generation.  This procedure is usually modeled by bringing together two or more quantum systems initially in a product state $\rho\otimes\sigma$ and then applying an entangling joint operation.  However, using only incoherent operations, this will not be possible unless either $\rho$ or $\sigma$ already possesses coherence.   The reason is that when $\rho\otimes\sigma$ is an incoherent bipartite state, any incoherent operation acting on both systems will leave the joint state incoherent (and hence unentangled).  On the other hand, if the joint state is $\ket{+}\ket{0}$, with $\ket{\pm}=\sqrt{1/2}(\ket{0}\pm\ket{1})$, then an application of CNOT yields the entangled state $\sqrt{1/2}(\ket{00}+\ket{11})$.  This example reveals that coherence, or at least coherent-generating operations, is a pre-requisite for producing entanglement.  In fact, as Streltsov \textit{et al.}~have shown \cite{Streltsov-2015a}, every coherent state can be used for the generation of entanglement in a manner similar to this example.

Notice that the transformation $\ket{+}\ket{0}\to\sqrt{1/2}(\ket{00}+\ket{11})$ requires performing an entanglement-generating incoherent operation.  
To capture both coherence and entanglement in a common resource-theoretic framework, one must modify the scenario by adopting the ``distant lab'' perspective in which two or more parties share a quantum system but they are spatially separated from one another \cite{Plenio-2007a, Horodecki-2009a}.  In this setting, entanglement cannot be generated between the parties and it becomes another resource in play.  When the constraint of locality is added to the incoherent framework, the allowable operations for Alice and Bob are then \textit{local incoherent operations and classical communication} (LIOCC).  The hybrid coherence-entanglement theory described here is similar in spirit to previous work on the locality-restricted resource theories of purity \cite{Oppenheim-2002a, Horodecki-2003b, Oppenheim-2003a, Horodecki-2005a} and asymmetry \footnote{The resource theory of asymmetry \cite{Bartlett-2007a, Gour-2008a, Gour-2009a, Marvian-2013a, Marvian-2014a} under locality constraints has been studied in Refs. \cite{Vaccaro-2008a, White-2009a}.  Note that the precise connection between coherence and asymmetry is rather subtle, since in the latter one can allow for decoherence-free subspaces when taking tensor products.  For instance, if $U(1)$ is the local symmetry for Alice, then when considering two copies of her system, the coherence-generating operation $\ket{01}^{A_1A_2}\to 1/\sqrt{2}(\ket{01}^{A_1A_2}+\ket{10}^{A_1A_2})$ is allowed.  See \cite{Chitambar-2016a} and \cite{Marvian-2016a} for more details.}.  We admittedly do not point to a specific biological or thermodynamic process as motivation for studying LIOCC -- although, one could envision potential physical applications in certain coherence-enhanced transport networks where the nodes interact through classical signaling.  Rather, we promote LIOCC as the natural setting to explore the interplay between coherence and entanglement as resource primitives in quantum information theory.  For example, how much local coherence and shared entanglement do Alice ($A$) and Bob ($B$) need to prepare a particular bipartite state $\rho^{AB}$ using LIOCC (Fig.~\ref{Fig:formation-distillation} (a))?  Conversely, how much coherence and entanglement can be distilled from a given state $\rho^{AB}$ using LIOCC (Fig.~\ref{Fig:formation-distillation} (b))?  The latter task can also be seen as type of collaboative randomness distillation, where Alice and Bob work together to generate local sources of genuine randomness for each other \cite{Yuan-2015a}.

\begin{figure}[t]
\includegraphics[scale=.70]{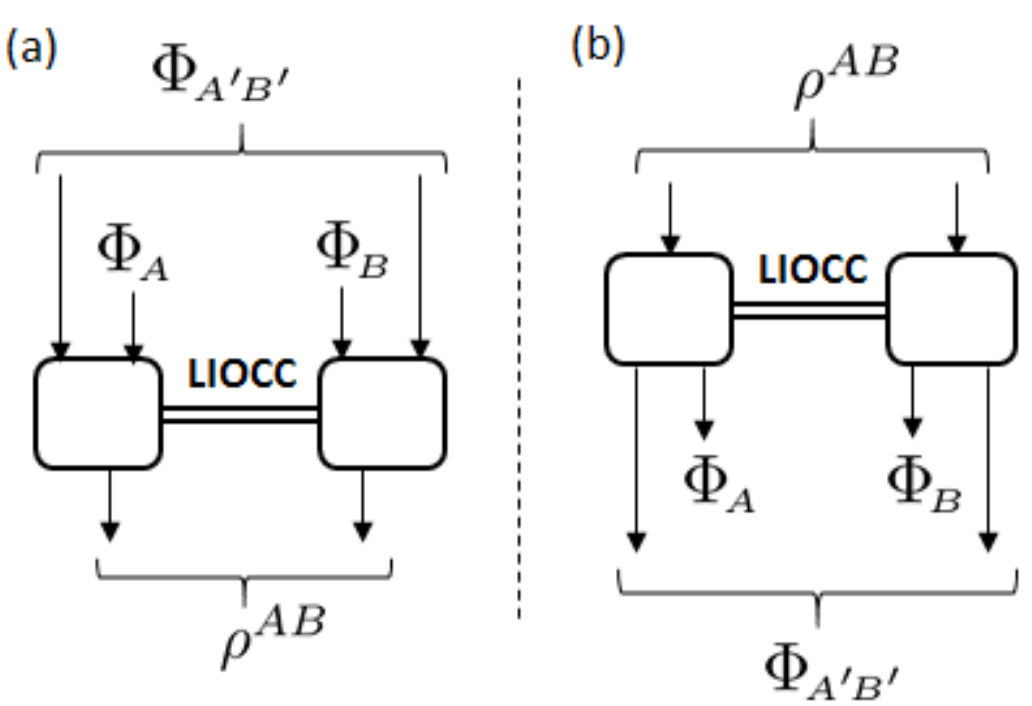}
\caption{\label{Fig:formation-distillation} (a) An LIOCC formation protocol asymptotically generates an arbitrary state $\rho^{AB}$ from an initial supply of local coherent bits ($\Phi_A/\Phi_B$) and shared entanglement bits ($\Phi_{A'B'}$).   (b) An LIOCC dilution protocol performs the reverse transformation.}
\end{figure}

Our main results are the following.  (1) We completely characterize the achievable coherence-entanglement rate region for the task of asymptotically generating some pure state $\ket{\Psi}^{AB}$ (Theorem \ref{thm1}).  (2) We introduce a new LIOCC monotone that combines both coherence and entanglement measures (Theorem \ref{thm4}), and we show it quantifies the optimal rate in which Alice and Bob can simultaneously distill local coherence from a pure state.  (3)  We identify an achievable rate region for the coherence-entanglement distillation of a pure state and show optimality at almost all corner points (Theorem \ref{thm5}).  
(4)  In analogy to Refs.~\cite{Oppenheim-2002a, Horodecki-2003b, Oppenheim-2003a, Horodecki-2005a}, we introduce and compute for pure states the nonlocal coherence deficit and the LIOCC coherence deficit (Eqns.~\eqref{Eq:nonlocal-deficit}--\eqref{Eq:LIOCC-deficit}).  (5)  We show that LIOCC operations alone are sufficient to decide whether entanglement can be distilled from a mixed state using general LOCC.

Let us begin by briefly describing the theory of bipartite coherence in more detail.  Assigned to both Alice and Bob's system is a particular basis called their incoherent basis.  We denote Alice's incoherent basis by $\{\ket{x}^A\}_{x=0}^{d_A-1}$ and Bob's incoherent basis by $\{\ket{y}^B\}_{y=0}^{d_B-1}$ so that the incoherent basis for their joint system $\mc{H}^A\otimes\mc{H}^B$ is $\{\ket{x}^A\ket{y}^B\}_{x,y=0}^{d_A-1,d_B-1}$.  Then any bipartite state belongs to the set of incoherent states $\mc{I}$ iff it has the form
\begin{equation}
\label{Eq:Incoherent}
\sigma^{AB}=\sum_{xy}p_{xy}\op{x}{x}^A\otimes\op{y}{y}^B.
\end{equation}
Following the framework of Baumgratz \textit{et al.}~\cite{Baumgratz-2014a}, a local incoherent operation for Alice is given by a complete set of Kraus operators $\{K_\alpha\}_{\alpha}$ such that $(K_\alpha\otimes\mbb{I}^B)\rho^{AB}(K_\alpha\otimes\mbb{I}^B)^\dagger/tr[K_\alpha K_\alpha^\dagger\otimes\mbb{I}^B\rho^{AB}]\in\mc{I}$ for all $\rho^{AB}\in\mc{I}$.  If ever she introduces a local ancilla system $\mc{H}^{A'}$, the incoherent basis for this additional system is labeled in the same way $\{\ket{x}^{A'}\}_{x=0}^{d_{A'}-1}$.  Analogous statements characterize the notion of incoherent operations on Bob's system.  In the LIOCC setting, Alice and Bob take turns performing local incoherent operations and sharing their measurement data over a classical communication channel.  

The canonical resource states in the bipartite LIOCC framework are the maximally coherent bits (CoBits), $\ket{\Phi_A}:=\sqrt{1/2}(\ket{0}^A+\ket{1}^A)$ and $\ket{\Phi_B}:=\sqrt{1/2}(\ket{0}^B+\ket{1}^B)$ for Alice and Bob's systems respectively \cite{Baumgratz-2014a}, as well as the entangled state $\ket{\Phi_{AB}}:=\sqrt{1/2}(\ket{00}+\ket{11})$, which we will call the maximally coherent entangled bit (eCoBit).  Notice that unlike entanglement theory, only those bipartite states related to $\ket{\Phi_{AB}}$ by an incoherent local unitary transformation can be regarded as equivalent to $\ket{\Phi_{AB}}$.  For example, as we will see below, one eCoBit cannot be incoherently transformed into the state $\sqrt{1/2}(\ket{0+}+\ket{1-})$, even asymptotically.

We now describe the primary tasks studied in this paper, which can be seen as the resource-theoretic tasks recently analyzed by Winter and Yang in Ref.~\cite{Winter-2015a} but now with additional locality constraints.  All of the detailed proofs can be found in the Supplemental Material, and here we just present the results.  Let us begin with the problem of asymptotic state formation shown in Fig.~\ref{Fig:formation-distillation} (a).  A triple $(R_A,R_B,E^{co})$ is an achievable \textit{coherence-entanglement formation triple} for the state $\rho^{AB}$ if for every $\epsilon>0$ there exists an LIOCC operation $\mc{L}$ and integer $n$ such that
\begin{equation}
\mc{L}\left(\Phi_A^{\otimes\lceil n(R_A+\epsilon)\rceil}\otimes\Phi_B^{\otimes\lceil n(R_B+\epsilon)\rceil}\otimes\Phi_{A'B'}^{\otimes\lceil n(E^{co}+\epsilon)\rceil}\right)\overset{\epsilon}{\approx}\rho^{\otimes n}.\notag
\end{equation}
Dual to the task of formation is resource distillation, as depicted in Fig.~\ref{Fig:formation-distillation} (b).  A triple $(R_A,R_B,E^{co})$ is an achievable \textit{coherence-entanglement distillation triple} for $\rho^{AB}$ if for every $\epsilon>0$ there exists an LIOCC operation $\mc{L}$ and integer $n$ such that
\begin{equation} 
\mc{L}(\rho^{\otimes n})\overset{\epsilon}{\approx}\Phi_A^{\otimes\lfloor n(R_A-\epsilon)\rfloor}\otimes\Phi_B^{\otimes\lfloor n(R_B-\epsilon)\rfloor}\otimes\Phi_{AB}^{\otimes\lfloor n(E^{co}-\epsilon)\rfloor}.\notag\
\end{equation}
As we are dealing with asymptotic transformations, we should expect the optimal rate triples to be given by entropic quantities.  Recall that for a bipartite state $\omega^{AB}$, the von Neumann entropy of, say, Alice's reduced state $\omega^{A}$ is given by $S(A)_{\omega}=-\tr[\omega^A\log\omega^A]$.  The quantum mutual information of $\omega^{AB}$ takes the form $I(A:B)_\omega:=S(A)_\omega-S(A|B)_{\omega}$, where $S(A|B)_{\omega}:=S(AB)_{\omega}-S(B)_{\omega}$.  For a pure state $\ket{\Psi}^{AB}$, the entropy of entanglement $\mathrm{E}(\Psi):=S(A)_{\Psi}=S(B)_{\Psi}$ is the unique measure of entanglement in the asymptotic regime \cite{Popescu-1997a}, and it can be generalized to mixed states as the entanglement of formation $\mathrm{E}_F(\rho)$ \cite{Bennett-1996a}.  We will also be interested in these entropic quantities after sending our state $\omega^{AB}$ through the completely dephasing channel, $\Delta(\omega):=\sum_{xy}\op{xy}{xy}\omega\op{xy}{xy}$.  It will be convenient to think of $\Delta(\omega)$ as encoding random variables $XY$ having joint distribution $p(x,y)=\bra{xy}\Delta(\omega)\ket{xy}$.  For this reason, we follow standard convention and replace the labels $(A,B)\to (X,Y)$ when discussing a dephased state.   

Our first main result completely characterizes the achievable rate region for the LIOCC formation of bipartite pure states.
\begin{theorem}
\label{thm1}
For a pure state $\ket{\Psi}^{AB}$ the following triples are achievable coherence-entanglement formation rates
\begin{align}
(R_A,R_B,E^{co})&=\left(\;0,\;S(Y|X)_{\Delta(\Psi)}\;,\; S(X)_{\Delta(\Psi)}\right)\label{Eq:CostTriple1}\\
(R_A,R_B,E^{co})&=\left(S(X)_{\Delta(\Psi)},\;S(Y|X)_{\Delta(\Psi)},\; \mathrm{E}(\Psi)\right)\label{Eq:CostTriple2}\\
(R_A,R_B,E^{co})&=\left(0,\;0,\; S(XY)_{\Delta(\Psi)}\right)\label{Eq:CostTriple3}
\end{align}
as well as the points obtained by interchanging $A\leftrightarrow B$ in Eqns.~\eqref{Eq:CostTriple1} -- \eqref{Eq:CostTriple3}.  Moreover, these points are optimal in the sense that any achievable rate triple must satisfy (i) $E^{co}\geq \mathrm{E}(\Psi)$, (ii) $R_A+R_B\geq S(XY)_{\Delta(\Psi)}$, (iii) $R_B+E^{co}\geq S(XY)_{\Delta(\Psi)}$.
\end{theorem}

For a mixed state $\rho^{AB}$, a formation protocol can be constructed that achieves the average rates for any ensemble $\{p_k,\ket{\varphi_k}^{AB}\}$ such that $\rho=\sum_kp_k\op{\varphi_k}{\varphi_k}$ \cite{Bennett-1996a}.  
For instance, one can consider an ensemble whose average bipartite coherence attains the coherence of formation $C_F$ for $\rho$; i.e. it is an ensemble $\{p_k,\ket{\varphi_k}^{AB}\}$ for $\rho$ that minimizes $\sum_k p_k S(XY)_{\Delta(\varphi_k)}$ \cite{Yuan-2015a, Winter-2015a}.  Then for a mixed state $\rho$, the coherence rate sum $R_A+R_B$ of Eq.~\eqref{Eq:CostTriple2} can attain the coherence of formation $C_F(\rho)$.  In the global setting where Alice and Bob are allowed to perform joint operations across system $AB$, it has been shown that $C_F(\rho)$ quantifies the optimal coherence consumption rate for generating $\rho$ using global incoherent operations \cite{Winter-2015a}.  Our result then intuitively says that in the restricted LIOCC setting, the same coherence rate is sufficient to generate $\rho$, however they now need additional entanglement at a rate $\sum_k p_k \mathrm{E}(\varphi_k)$, where the ensemble $\{p_k,\ket{\varphi_k}^{AB}\}$ minimizes the average coherence of $\rho$.

The proof of Theorem \ref{thm1} uses two lemmas that may be of independent interest.  The first generalizes a result presented in Ref.~\cite{Baumgratz-2014a}, and the second is an incoherent version of Nielsen's Majorization Theorem \cite{Nielsen-1999a}.
\begin{lemma}
\label{lem2}
An arbitrary $d\times d$ unitary operator $U$ can be performed on a system using incoherent operations and $\lceil \log d\rceil$ CoBits.
\end{lemma}
\begin{lemma}
\label{lem:majorization}
Suppose $\ket{\psi}^{AB}$ and $\ket{\phi}^{AB}$ have reduced density matrices that are diagonal in the incoherent bases for both parties and both states.  Then $\ket{\psi}\to\ket{\phi}$ by LIOCC iff the squared Schmidt coefficients of $\ket{\phi}$ majorize those of $\ket{\psi}$.
\end{lemma}

Next, we introduce a new LIOCC monotone and provide its operational interpretation.  To do so, we recall the recently studied task of \textit{assisted} coherence distillation, which involves one party helping another distill as much coherence as possible through general quantum operations performed on the helper side and incoherent operations performed on the distillation side \cite{Chitambar-2015a}.  For a given state $\rho^{AB}$, the optimal asymptotic rate of coherence distillation on Bob's side when Alice helps is denoted by $C_a^{A|B}(\rho^{AB})$.  When the roles are switched, the optimal asymptotic rate is denoted by $C_a^{B|A}(\rho^{AB})$.  It was shown in Ref. \cite{Chitambar-2015a} that $C_a^{A|B}(\rho^{AB})=S(Y)_{\Delta(\Psi)}$ and $C_a^{B|A}(\rho^{AB})=S(X)_{\Delta(\Psi)}$.  With these quantities in hand, we define for a bipartite pure state $\ket{\Psi}^{AB}$ the function
\begin{align}
C_{\mc{L}}(\Psi)&:=C^{A|B}_a(\Psi)+C^{B|A}_a(\Psi)-\mathrm{E}(\Psi)\notag\\
&=S(X)_{\Delta(\Psi)}+S(Y)_{\Delta(\Psi)}-\mathrm{E}(\Psi).
\end{align}
Its extension to mixed states can be defined by a convex roof optimization \cite{Vidal-2000a}: $C_{\mc{L}}(\rho^{AB})=\inf_{\{p_k,\ket{\varphi_k}^{AB}\}}\sum_k p_kC_{\mc{L}}(\varphi_k^{AB})$ for which $\rho^{AB}=\sum_kp_k\op{\varphi_k}{\varphi_k}$. 
\begin{theorem}
\label{thm4}
The function $C_{\mc{L}}$ is an LIOCC monotone.
\end{theorem}
\noindent We note that this is the first monotone of its kind since it behaves monotonically under LIOCC, but not general LOCC or even under LQICC, the latter being an operational class in which only one of the parties is required to perform incoherent operations (as opposed to LIOCC where \textit{both} parties must perform incoherent operations) \cite{Chitambar-2015a}.  

Using the monotonicity of $C_{\mc{L}}$, we are able to derive tight upper bounds on coherence distillation rates.  
\begin{theorem}
\label{thm5}
For a pure state $\ket{\Psi}^{AB}$ the following triples are achievable coherence-entanglement distillation rates
\begin{align}
(R_A,R_B,E^{co})&=\left(S(X)_{\Delta(\Psi)}-\mathrm{E}(\Psi),\;S(Y)_{\Delta(\Psi)},\;0\right)\label{Eq:DistTriple1}\\
(R_A,R_B,E^{co})&=\left(0,\;S(Y|X)_{\Delta(\Psi)},\;I(X:Y)_{\Delta(\Psi)}\right),\label{Eq:DistTriple4}
\end{align}
as well as the points obtained by interchanging $A\leftrightarrow B$ in Eqn.~\eqref{Eq:DistTriple1} and \eqref{Eq:DistTriple4}.  Moreover, these points are optimal in the sense that any achievable rate triple must satisfy (i) $R_A+R_B\leq C_{\mc{L}}(\Psi)$ and (ii) $R_B+E^{co}\leq S(Y)_{\Delta(\Psi)}$.
\end{theorem}
\noindent This theorem endows $C_{\mc{L}}$ with the operational meaning of quantifying how much local coherence can be simultaneously distilled from a pure state.  For a state $\ket{\Psi}$ the maximum that Alice can help Bob distill coherence is $C_a^{A|B}$ while the maximum that Bob can help Alice is $C_a^{B|A}$.  Evidently, they cannot both simultaneously help each other at these optimal rates.  Instead, they are bounded away from simultaneous optimality at a rate equaling their shared entanglement.  

It is still unknown the precise range of achievable distillation triples $(R_A, R_B, E_{max}^{co})$, where $E_{max}^{co}$ is the maximum eCoBit distillation rate.  While we are able to prove that $E_{max}^{co}$ is the regularized version of $I(X:Y)_{\Delta(\Psi)}$ optimized over all LIOCC protocols, we have no single-letter expression for this rate nor do we know the achievable local coherence rates for optimal protocols.  

A natural question is whether $E_{max}^{co}(\Psi)=\mathrm{E}(\Psi)$.  While this question remains open, we can show that $\mathrm{E}(\Psi)$ is achievable if the Schmidt basis of the final state need not be incoherent.  More precisely, we say a number $R$ is an achievable LIOCC entanglement distillation rate if for every $\epsilon>0$, there exists an LIOCC protocol $\mc{L}$ acting on $n$ copies of $\Psi$ such that $\mc{L}(\Psi^{\otimes n})\overset{\epsilon}{\approx} \Lambda_d$,
where $\Lambda_d$ is a $d\otimes d$ maximally entangled pure state (i.e. $\Lambda^A=\Lambda^B=\mbb{I}/d$) with $\frac{1}{n}\log d>R-\epsilon$.  The largest achievable distillation rate will be denoted by $E^{LIOCC}_D(\Psi)$.
\begin{theorem}
\label{thm6} 
$E^{LIOCC}_D(\Psi)=\mathrm{E}(\Psi)$.
\end{theorem}

It is interesting to compare the coherence distillation rates using incoherent operations under different types of locality constraints.  In Refs.~\cite{Oppenheim-2002a, Horodecki-2003b, Oppenheim-2003a, Horodecki-2005a}, similar comparisons were made in terms of purity (or work-information) extraction.  Let $C_D^{Global}$, $C_D^{LIOCC}$, and $C_D^{LIO}$ denote the optimal rate sum $R_A+R_B$ of local coherence distillation using global incoherent operations, LIOCC, and local incoherent operations (with no classical communication), respectively.  In complete analogy to \cite{Oppenheim-2002a, Horodecki-2003b, Oppenheim-2003a, Horodecki-2005a}, we define the \textit{nonlocal coherence deficit} of a bipartite state $\rho^{AB}$ as $\delta(\rho^{AB})=C_D^{Global}(\rho^{AB})-C_D^{LIOCC}(\rho^{AB})$ and the \textit{LIOCC coherence deficit} as $\delta_c(\rho^{AB})=C_D^{LIOCC}(\rho^{AB})-C_D^{LIO}(\rho^{AB})$.  Intuitively, the quantity $\delta(\rho^{AB})$ quantifies the coherence in a state that can only be accessed using nonlocal incoherent operations.  Likewise, $\delta_c(\rho^{AB})$ gives the coherence in $\rho^{AB}$ that requires classical communication to be obtained.  The results of Winter and Yang imply that $C_D^{Global}(\Psi)=S(XY)_{\Delta(\Psi)}$ and $C_D^{LIO}(\Psi)=S(X)_{\Delta(\Psi)}+S(Y)_{\Delta(\Psi)}-2\mathrm{E}(\Psi)$ for a bipartite pure state $\ket{\Psi}^{AB}$ \cite{Winter-2005a}.  Combined with Theorem \ref{thm5}, we can compute the two coherence deficits for pure states:
\begin{align}
\delta(\Psi)&=\mathrm{E}(\Psi)-I(X:Y)_{\Delta(\Psi)}\label{Eq:nonlocal-deficit}\\
\delta_c(\Psi)&=\mathrm{E}(\Psi).\label{Eq:LIOCC-deficit}
\end{align}
It is curious that the entanglement $\mathrm{E}(\Psi)$ quantifies the coherence gain unlocked by \textit{classical} communication. But note that a similar phenomenon exists in the resource theory of purity.  Namely, the quantum deficit $\overline{\delta}(\Psi)$ and classical deficit $\overline{\delta}_c(\Psi)$ measure the analogous differences in local purity distillation by so-called ``closed operations'' (CO), and they are given by $\overline{\delta}(\Psi)=\overline{\delta}_c(\Psi)=\mathrm{E}(\Psi)$ \cite{Oppenheim-2002a, Horodecki-2003b}.  For the task of distilling CoBits, every protocol using incoherent operations can be seen as one using closed operations by accounting for all ancilla systems at the start of protocol \footnote{With sufficient ancilla systems, every incoherent operation can be modeled using a projective measurement on a larger system and classical processing \cite{Chitambar-2016a}.  However, since the target states are pure, one does not need to classically process measurement outcomes or discard any subsystems in a successful distillation protocol, and thus the whole procedure can be done on a ``closed'' system \cite{Horodecki-2005a}}.  However, closed operations allow for arbitrary unitary rotations, which are forbidden in coherence theory.  The term $I(X:Y)_{\Delta(\Psi)}$ in $\delta(\Psi)$ identifies precisely the basis dependence in coherence theory and shows how this decreases the nonlocal coherence deficit $\delta(\Psi)$ relative to $\overline{\delta}(\Psi)$.  On the other hand, there is evidently no basis dependency in the LIOCC classical deficit $\delta_c(\Psi)$ and it is equivalent to $\overline{\delta}_c(\Psi)$.

Although our distillation results so far have only applied to pure states, we can deduce a very general result concerning the distillability of mixed states.
\begin{theorem}
\label{thm7}
A mixed state $\rho^{AB}$ has (LOCC) distillable entanglement if and only if entanglement can be distilled using LIOCC.
\end{theorem}
The proof of this theorem is actually quite simple and uses the fact that an arbitrary quantum operation can be simulated using incoherent operations and CoBits (Lemma \ref{lem2}).  In Ref.~\cite{Chitambar-2015a} it was shown how local coherence can always be distilled for both Alice and Bob from multiple copies of every entangled states using LIOCC.  Hence for a sufficiently large number of any distillable entangled state $\rho^{AB}$, Alice and Bob first distill sufficient local coherence using LIOCC, and then they simulate the LOCC protocol which distills entanglement.

As shown in Ref.~\cite{Horodecki-1998a}, a state $\rho$ has distillable entanglement iff for some $k$ there exists rank two operators $A$ and $B$ such that the (unnormalized) state $A\otimes B\rho^{\otimes k} A\otimes B$ is entangled.  By Theorem \ref{thm5} and following the same argumentation of Ref.~\cite{Horodecki-1998a}, we can further require that the $A$ and $B$ are incoherent operators; that is, they have the form $A=\op{0}{\alpha_0}+\op{1}{\alpha_1}$ and $B=\op{0}{\beta_0}+\op{1}{\beta_1}$ where $\Delta(\alpha_0):=\Delta(\op{\alpha_0}{\alpha_0})$ is orthogonal to $\Delta(\alpha_1):=\Delta(\op{\alpha_1}{\alpha_1})$, and likewise for $\Delta(\beta_0):=\Delta(\op{\beta_0}{\beta_0})$ for $\Delta(\beta_1):=\Delta(\op{\beta_1}{\beta_1})$.  We are thus able to add an additional condition to the distinguishability criterion of Ref.~\cite{Horodecki-1998a}.
\begin{corollary}
\label{cor8}
A bipartite state $\rho$ has distillable entanglement iff for any pair of orthonormal local bases $\mc{B}_A=\{\ket{x}^A\}$ and $\mc{B}_B=\{\ket{y}^B\}$ there exists some $k$ and projectors $P_A=\op{\alpha_0}{\alpha_0}+\op{\alpha_1}{\alpha_1}$ and $P_B=\op{\beta_0}{\beta_0}+\op{\beta_1}{\beta_1}$ such that \begin{enumerate}
\item $(P_A\otimes P_B)\rho^{\otimes k} (P_A\otimes P_B)$ is entangled,
\item $\tr[\Delta_A(\alpha_0)\Delta_A(\alpha_1)]=\tr[\Delta_B(\beta_0)\Delta_B(\beta_1)]=0$,
\end{enumerate} where $\Delta_Z$ is the completely dephasing map in the basis $\mc{B}^{\otimes k}_Z$.
\end{corollary} 



\textit{Conclusion:}  In this letter, we have investigated the relationship between entanglement and coherence in the framework of local incoherent operations and classical communication.   The findings of this study suggest that indeed entanglement and coherence are closely linked resources.  For instance, Theorem \ref{thm5} shows that the entanglement of a state plays a crucial role in limiting the amount of coherence that can be distilled from a state, a result highly reminiscent of the complementarity between local and nonlocal information studied in Ref.~\cite{Oppenheim-2003a}.  In a similar spirit, Theorem \ref{thm7} shows that entanglement distillability can be studied through the lens of coherence theory.  This latter result seems somewhat remarkable since despite coherence being a basis-dependent resource, its resource-theoretic analysis can be used to draw conclusions about entanglement, a basis-independent resource.  Future work will be conducted to see whether the strengthened distillability criterion of Corollary \ref{cor8} can be useful in the long-standing search for NPT bound entanglement.

Finally, we would like to comment on the particular type of incoherent operations studied in this letter.  As noted in the introduction, there have been various proposals for the ``free'' class of operations in a resource theory of coherence.  This letter has adopted the incoherent operations (IO) of Baumgratz \textit{et al.}~\cite{Baumgratz-2014a}, where each Kraus operator in a measurement just needs to be incoherence-preserving.  While the class IO has drawbacks in terms of formulating a full physically consistent resource theory of coherence \cite{Yadin-2015a, Chitambar-2016a}, it nevertheless seems unlikely that the results of this letter would remain true if other operational classes were considered.  For example, the strictly incoherent operations (SIO) proposed by Yadin \textit{et al.}~are unable to convert one eCoBit into a CoBit \cite{Yadin-2015a}.  Thus, we believe that the interesting connections between IO coherence theory and entanglement demonstrated in this letter make a positive case for why IO is important in quantum information theory, independent of any other motivation.  In fact, one could even put coherence aside and view LIOCC as just being a simplified subset of LOCC.  As we have shown here, nontrivial conclusions about entanglement can indeed be drawn by studying LOCC from ``the inside.''  This approach is somewhat dual to the standard practice of studying LOCC using more general separable operations (SEP), the chain of inclusions being LIOCC $\subset$ LOCC $\subset$ SEP.  Interesting future work would be to consider more general connections between coherence non-generating and entanglement non-generating operations.

During preparation of this manuscript, we learned of work by Streltsov and co-authors who have also initiated a study into local incoherent operations and classical communication \cite{Streltsov-2015b}.  

\medskip

\noindent\textit{\textbf{Acknowledgments}}

We thank Alex Streltsov for fruitful exchanges on the topic of coherence distillation.  EC is supported by the National Science Foundation (NSF) Early CAREER Award No. 1352326. MH is supported by an ARC Future Fellowship under Grant FT140100574.

\bibliography{CoherenceBib}

\onecolumngrid

\bigskip

\bigskip

\begin{center}
{\huge Supplemental Material}
\end{center}

\section{Preliminaries}

\subsection{Distance Measures}

The distance measure used in this paper is based on the trace norm, which for an operator $A$ is $\norm{A}:=\tr|A|=\tr\sqrt{A^\dagger A}$.  The trace distance for two states $\rho$ and $\sigma$ is $D_{tr}(\rho,\sigma)=\tfrac{1}{2}\norm{\rho-\sigma}$, and we will write $\rho\overset{\epsilon}{\approx}\sigma$ to indicate $D_{tr}(\rho,\sigma)\leq\epsilon$.  The fidelity of two states - given by $F(\rho,\sigma)=\tr\sqrt{\sqrt{\rho}\sigma\sqrt{\rho}}$ - can be related to the trace distance by \cite{Fuchs-1999a}
\[1-F(\rho,\sigma)\leq D_{tr}(\rho,\sigma)\leq\sqrt{1-F(\rho,\sigma)^2}.\]
When $\sigma$ is pure, the lower bound can be improved:
\begin{equation}
1-F(\rho,\op{\varphi}{\varphi})^2\leq D_{tr}(\rho,\op{\varphi}{\varphi}).
\end{equation}
We reference Fannes' inequality, which provides a bound on entropy difference in terms of the trace distance.
\begin{lem}[Fannes-Audenaert Inequality \cite{Fannes-1973a, Audenaert-2007a}]
\label{lem:Fannes}
For density matrices $\rho$ and $\sigma$ acting on a $d$-dimensional space, 
\begin{equation}
|S(\rho)-S(\sigma)|\leq D_{tr}(\rho,\sigma)\log(d-1)+h\left(D_{tr}(\rho,\sigma)\right),
\end{equation}
where $h(x)=-x\log x-(1-x)\log(1-x)$.
\end{lem} 
\noindent Finally, we will need Winter's gentle measurement lemma.
\begin{lem}[Gentle Measurements \cite{Winter-1999a}] 
\label{Lem:gentle}
For $\rho\geq 0$ and $\tr\rho\leq 1$, suppose $0\leq X\leq\mbb{I}$ and $\tr\rho X\geq 1-\epsilon$.  Then $||\rho-\sqrt{X}\rho\sqrt{X}||\leq\sqrt{8\epsilon}$.
\end{lem}

\subsection{Types, Typical Sequences, Channel Coding}

In what follows, we assume that random variables $X,Y,\cdots$ take on values $x,y,\cdots$ from sets $\mc{X},\mc{Y},\cdots$. Probability distributions will be denoted by $p$ or $q$.  For $n$ identical and independently distributed (i.i.d.) events each with outcome distribution $p$, the distribution over the sequence of events is denoted by $p^n$.  See Refs. \cite[Chapter 2]{Csiszar-2011a} and \cite{Devetak-2004b} for a comprehensive presentation of the following concepts.

The \textbf{type} of a sequence $x^n\in\mc{X}^n$ is the distribution $p_{x^n}$ over $\mc{X}$ defined by
\[p_{x^n}(a):=\frac{1}{n}N(a|x^n)\qquad \forall a\in\mc{X},\]
where $N(a|x^n)$ is the number of occurrences of the symbol $a\in\mc{X}$ in the sequence $x^n$.
For a given distribution $p$, the collection of all sequences having type $p$ is called the \textbf{type class} of $p$ and is denoted by $T_p^n$.  A distribution $p$ is said to be an \textbf{empirical type} (for some $n\in\mbb{N}$) if $T_p^n$ is nonempty.

A sequence $x^n$ is said to be \textbf{$\delta$-typical} (or just typical) w.r.t. distribution $p$ if 
\[\left|\frac{1}{n}N(a|x^n)-p(a)\right|<\delta\qquad\forall a\in\mc{X}.\]
The set of all $\delta$-typical sequences will be denoted by $T^n_{[p]_\delta}$.  
Note that the set $T^n_{[p]_\delta}$ is the union of empirical type classes, and hence we will say that a distribution $q$ is typical w.r.t. $p$ if it is an empirical type with $T^n_q\subset T^n_{[p]_\delta}$.  

Three standard properties of typicality are the following.  First, for $n$ i.i.d. samples of $\mc{X}$ according to distribution $p$, 
\begin{equation}
\label{Eq:typicalProb}
p^n\left(T^n_{[p]_\delta}\right):=\pr[x^n\in T^n_{[p]_\delta}]\geq 1-\epsilon
\end{equation}
for any $\epsilon,\delta>0$ and $n$ sufficiently large.  Second, let $X$ be a random variable having distribution $p$ and entropy $H(X):=-\sum_{a\in\mc{X}}p(a)\log p(a)$.  Then the size of $T^n_{[p]_\delta}$ can be related to $H(X)$ as   
\begin{equation}
\bigg|\frac{1}{n}\log|T^n_{[p]_{\delta}}|-H(X)\bigg| \leq \epsilon
\end{equation}
 for any $\epsilon,\delta>0$ and $n$ sufficiently large.  Third, if $q\in T^n_{[p]_\delta}$ then (for $\delta< (2|\mc{X}|)^{-1}$)
\begin{equation}
\label{Eq:TypicalTypeNumber}
(n+1)^{-|\mc{X}|}2^{n(H(X)-\tau(\delta))}\leq |T^n_{q}|\leq 2^{n(H(X)+\tau(\delta))},
\end{equation}
where $\tau(\delta)=-\delta|\mc{X}|\log\delta$.  Note that $\tau(\delta)$ is a function increasing in $\delta$ for the range $0<\delta< (2|\mc{X}|)^{-1}$.  Eq. \eqref{Eq:TypicalTypeNumber} follows from an application of Lemma \ref{lem:Fannes} to the inequality
\begin{equation}
(n+1)^{-|\mc{X}|}2^{nH(X)}\leq |T^n_{p}|\leq 2^{nH(X)}
\end{equation}

Moving to the quantum setting, for a quantum system $\mc{H}$, we will assume throughout that the computational basis $\{\ket{x}\}$ is the incoherent basis.  For an empirical type $p$, the corresponding \textbf{type projector} acting on $\mc{H}^{\otimes n}$ is given by
\[\Pi_p=\sum_{x^n\in T^n_p}\op{x^n}{x^n}.\]
We restrict attention to CQ-channels $\mc{W}_{\rm CQ}:\op{x}{x}^X\to\rho_{x}^B$, which map each element $\ket{x}$ to a density matrix $\rho_x^B$ acting on $\mc{H}$.  Note that CQ channels generalize classical channels.  If we are given a set of transition probability $p(y|x)$ characterizing a classical channel, then the corresponding CQ channel is $\op{x}{x}^X\mapsto\rho_x^Y$, where $\rho_x^Y=\sum_{y\in\mc{Y}}p(y|x)\op{y}{y}$.    We will refer to channels of this form as CC channels $\mc{W}_{\rm CC}$.

When a distribution $p$ is given over $\mc{X}$, we associate a quantum ensemble $\{p(x),\rho_x^B\}_{x\in\mc{X}}$ with the channel $\mc{W}_{\rm CQ}$, as well as a classical-quantum state $\rho^{XB}$:
\[
\rho^{XB} = \sum_{x\in\mc{X}} p(x) \proj{x}^X \otimes \rho_x^B.
\]
The mutual information of the classical-quantum state $\rho^{XB}$ is the so-called Holevo quantity and denoted by $I(X:B)=S(\sum_x p(x)\rho_x^B)-\sum_x p(x)S(\rho_x^B)$.  For CC channels, the associated joint state is fully incoherent:
\[\rho^{XY}=\sum_{x\in\mc{X},y\in\mc{Y}}p(x,y)\op{x}{x}^X\otimes\op{y}{y}^Y,\]
with mutual information $I(X:Y)=H(X)+H(Y)-H(XY)$.

Recall that an $(n,\epsilon)$ code of size $C$ for a CQ channel $\mc{W}_{\rm CQ}:\op{x}{x}^X\to\rho_{x}^B$ is a sequence of codewords $(U^{(c)})_{c=1}^C$ with $U^{(c)}\in \mc{X}^n$ and a POVM $\{D_c\}_{c=1}^C$ acting on $\mc{H}^{\otimes n}$ such that
\begin{equation}
\label{Eq:epsilon-good}
\frac{1}{C}\sum_{c=1}^C \tr[\rho^n_{U^{(c)}}D_{c}]>1-\epsilon,
\end{equation}
where if $U^{(c)}=x^n$, then $\rho^n_{U^{(c)}}:=\rho_{x_1}^B\otimes\cdots\otimes\rho_{x_n}^B \in \mc{H}^{\otimes n}$.

\subsection{Coding Theorems}

We now introduce the information-theoretic machinery that provides the foundation for our coding schemes.  The following is adopted from the work of Devetak and Winter in Ref.~\cite{Devetak-2005a}.  For an empirical type $p$ over $\mc{X}$, let $(U^{(lc)})$ be an i.i.d.~sequence of random variables obtained by sampling from $T^n_p$ uniformly, with $l=1,\cdots, L$ and $c=1,\cdots,C$.  We consider the following events:
\begin{itemize}
\item \textit{$\epsilon$-evenness}:  For all $x^n\in T^n_p$,
\begin{equation}
(1-\epsilon)\frac{LC}{|T^n_p|}\leq \sum_{lc}\mathbbm{1}_{U^{(lc)}}(x^n)\leq (1+\epsilon)\frac{LC}{|T^n_p|},
\end{equation}
where $\mathbbm{1}_{U^{(lc)}}$ is the indicator function for whether $U^{(lc)}=x^n$.
\item \textit{$\{C_l\}_{l=1}^L$ are $(n,\epsilon)$ codes}: For every $l=1,\cdots, L$ the codebook $C_l:=(U^{(lc)})_{c=1}^C$ forms an $(n,\epsilon)$ code for the channel $\mc{W}_{\rm CQ}:\op{x}{x}^X\to\rho_{x}^B$.
\end{itemize}
\begin{lem}[\cite{Holevo-1998a, Schumacher-1997a, Devetak-2005a}]
\label{lem:HSWcover}
Consider a CQ channel $\mc{W}_{\rm CQ}:\op{x}{x}^X\to\rho_{x}^B$ and a random variable $X$ with a distribution given by some empirical type $p$.  Let $(U^{(lc)})$ be an i.i.d.~sampling from $T^n_p$ with $l=1,\cdots, L$, $c=1,\cdots, C$, and $C_l=(U^{(lc)})_{c=1}^C$.  For every $\delta,\epsilon>0$ and $n$ sufficiently large, 
\begin{align}
&\pr\{\text{$\epsilon$-evenness}\}\geq 1-|\mc{X}|^n2^{-LC\tfrac{\epsilon^2}{2\ln 2|T^n_p|}},\notag\\
C\leq 2^{n(I(X:B)-\delta)}\quad\Rightarrow\quad&\pr\{\text{A fraction $1-2\epsilon$ of the $\{C_l\}_{l=1}^L$ are $(n,\epsilon)$ channel codes}\}\geq 1-2^{-L\tfrac{\epsilon^2}{4\ln 2}}.
\end{align}
\end{lem}

\begin{cor}
\label{cor:HSWcover}
Consider a CQ channel $\mc{W}_{\rm CQ}:\op{x}{x}^X\to\rho_x^B$ and let $X$ be a distribution given by some empirical type $p$.  For sufficiently large $n$ and $\delta<I(X:B)$, there exists a partition of $T^n_p$ such that a fraction $1-3\epsilon$ of the sequences in $T^n_p$ belong to an $(n,\epsilon)$ channel code $C_1,\cdots, C_L$, where $L=\lceil 2^{n(H(X)- I(X:B)+2\delta)}\rceil$ and each $C_l$ consists of $C=\lfloor 2^{n(I(X:B)-\delta)}\rfloor$ codewords.
\end{cor}
\begin{proof}
Since by the choices of $L$ and $C$,
\[
\frac{LC}{|T^n_p|}\geq  2^{n\delta}+2^{-n(H(X|B)+\delta)}\to\infty,
\]
then by Lemma \ref{lem:HSWcover} an i.i.d. sequence $(U^{(lc)})$ will satisfy 
\begin{align}
&\pr\{\text{$\epsilon$-evenness}\}\to 1\notag\\
&\pr\{\text{A fraction $1-2\epsilon$ of the $C_l$ are $(n,\epsilon)$ channel codes}\}\to 1\notag
\end{align}
 as $n\to\infty$.  Thus for sufficiently large $n$, there must exist families of $(n,\epsilon)$ codes $(U^{(lc)})\subset T^n_{p}$ for $\mc{W}_{\textrm{CQ}}$ that cover $T^n_{p}$ with a fraction $1-2\epsilon$ of these codes being $(n,\epsilon)$ channel codes.  The union of these codes will consist of $(1-2\epsilon)LC$ codewords, including multiplicities.  But by $\epsilon$-evenness, the number of distinct codewords in this union will be at least $\frac{(1-2\epsilon)LC}{(1+\epsilon)LC/|T^n_p|}>(1-3\epsilon)|T^n_p|$.  Indeed, $\epsilon$-evenness guarantees that each individual $x^n$ has multiplicity no more than $(1+\epsilon)\frac{LC}{|T^n_p|}$ among all the codebooks.  Hence, the number of distinct codewords is at least $\frac{(1-2\epsilon)LC}{(1+\epsilon)LC/|T^n_p|}$, and so at least a fraction $(1-3\epsilon)$ of the elements of $T^n_p$ are codewords for an $(n,\epsilon)$ code.
\end{proof}

With Corollary \ref{cor:HSWcover}, we are able to almost entirely cover each type class $T^n_{p}$ by $(n,\epsilon)$ channel codes having a constant rate $C$.  We will also be interested in decomposing $T^n_{p}$ into ``obfuscation'' sets.  The following covering lemma is presented in \cite{Winter-2015a}.
\begin{lem}
\label{lem:Cover}
Consider a CQ channel $\mc{W}_{\rm CQ}:\op{x}{x}^X\to\rho_{x}^B$ on a $d$-dimensional Hilbert space and a random variable $X$ with a distribution given by some empirical type $p$.  Let $\{\Omega_s\}_{s=1}^S$ be obtained by a uniform sampling without replacement from the set $\{\rho^n_{x^n}:x^n\in T^n_p\}$.  Define the average state
\[\sigma(p):=\frac{1}{|T_p^n|}\sum_{x^n\in T^n_p}\rho_{x^n}.\]
Then for every $\epsilon,\delta\in(0,1)$ and $n$ sufficiently large,
\begin{equation}
\pr\left\{\|\frac{1}{S}\sum_{s=1}^S \Omega_s-\sigma(p)\|\geq\epsilon\right\}\leq 2 d^n \exp\left(-S t^n\frac{\epsilon^2}{288\ln 2}\right),
\end{equation}
where $t=2^{-(I(X:B)+\delta)}$.
\end{lem}
\noindent We will say that a collection of states $\{\Omega_s\}_{s=1}^S$ corresponding to $S$ distinct sequences from $T^n_p$ is ``good'' if $\norm{\frac{1}{S}\sum_{s=1}^S \Omega_s-\sigma(p)}<\epsilon$.
\begin{cor}[\cite{Winter-2015a}]
\label{cor:Cover}
Let $X$ be a random variable with distribution given by some empirical type $p$.  For $n$ sufficiently large, there exists a partitioning of $T^n_{p}$ consisting of $M$ subsets $\{S_m\}_{m=1}^{M}$ each of size $S=\lceil 2^{n(I(X:B)+\delta)}\rceil$ (plus a remainder) such that a fraction $(1-\epsilon)$ of the subsets are good sets. 
\end{cor}
\begin{proof}
Consider a random partition of $T^n_{p}$ into $M$ blocks each of size $S$.  Note that each block is equivalently obtained by a uniform sampling of $S$ elements from $T^n_{p}$ without replacement.  Let $E_m$ be the random variable for which $E_m=1$ if the $m^{th}$ block is a good set and $E_m=0$ if it is not.  By applying Lemma \ref{lem:Cover}, $n$ can be taken sufficiently large so that the expectation of $E_m$ is greater than $1-\epsilon$.  Therefore for a random partition of $T^n_{p}$, the expected number of good sets across all blocks is given by $\langle \sum_{m=1}^M E_m\rangle= \sum_{m=1}^M\langle E_m\rangle >M(1-\epsilon)$.  Hence, there must exist at least one partition with $M(1-\epsilon)$ of the blocks being good.
\end{proof}

\section{Code Structure}

We combine Corollaries \ref{cor:HSWcover} and \ref{cor:Cover} to obtain the basic structure of both our distillation and formation codes.  A diagram is provided in Fig. \ref{Fig:Anatomy}.

\begin{figure}[t]
\includegraphics[scale=.75]{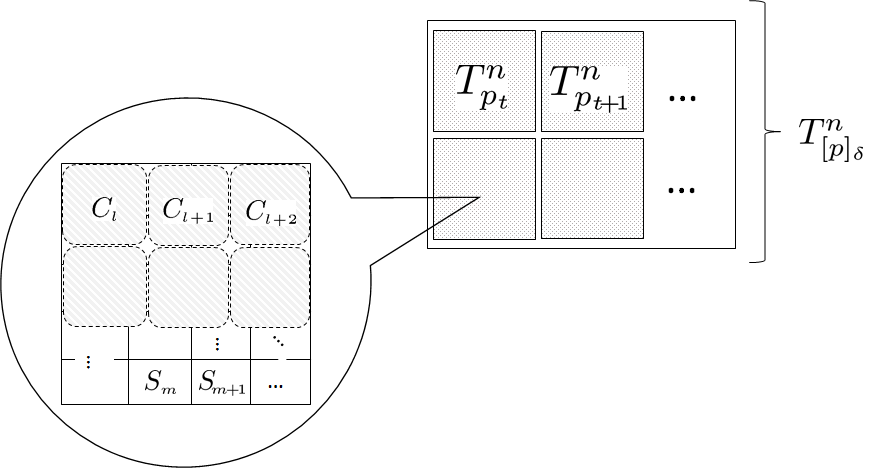}
\caption{\label{Fig:Anatomy} Code Depiction.  Our code involves first decomposing the typical set $T^n_{[p]_\delta}$ into its typical type classes $T^n_{p_t}$, where $p(x)$ is distribution given by $\ket{\Psi}^{AB}=\sum_{x\in\mc{X}}\sqrt{p(x)}\ket{x}^A\ket{\psi_x}^B$ for Alice's incoherent basis $\ket{x}^A$.  Each type class is further decomposed in three different ways.  The first involves a partitioning into obfuscation sets $S_m$ for which Bob's average state is roughly the same when restricting to these sets.  The other two decompositions involve partitioning $T^n_{p_t}$ into codebooks $C_l$ and $\overline{C_l}$ for the channels $\op{x}{x}\to\Delta(\op{\psi_x}{\psi_x})$ and $\op{x}{x}\to\op{\psi_x}{\psi_x}$ respectively.}
\end{figure}

Let $\ket{\Psi}^{AB}$ be an arbitrary bipartite state with 
\begin{equation}
\label{Eq:State}
\ket{\Psi}^{AB}=\sum_{x\in\mc{X}}\sqrt{p(x)}\ket{x}^A\ket{\psi_x}^B,
\end{equation}
where $\{\ket{x}^A\}$ and $\{\ket{y}^B\}$ denote the preferred bases with respect to which the incoherent operations are defined for Alice and Bob, respectively, and $\ket{\psi_x}^B=\sum_{y\in\mc{Y}}e^{i\theta_{y|x}}\sqrt{p(y|x)}\ket{y}^B$ are normalized but not necessarily orthogonal states of Bob.  Let $\mc{W}_{\textrm{CQ}}$ and $\mc{W}_{\textrm{CC}}$ be the CQ and CC channels given by 
\begin{align}
\mc{W}_{\textrm{CQ}}:\op{x}{x}^X&\to\psi_x^B\equiv \op{\psi_x}{\psi_x}^B\notag\\
\mc{W}_{\textrm{CC}}:\op{x}{x}^X&\to\Delta(\psi^B_x)=\sum_{y\in\mc{Y}}p(y|x)\op{y}{y}^Y.\notag
\end{align}
Let $X$ be the random variable taking on values from $\mc{X}$ according to the distribution $p(x)$ in Eq.~\eqref{Eq:State}.  In other words, $p(x)$ describes the distribution of outcomes when measuring $\Delta(\Psi^A)$ in the incoherent basis.  For a fixed $n$, the set of typical sequences $T^n_{[p]_\delta}$ is the union of typical types.  We will denote the typical types by $p_t$, for  $t=1,2,\cdots, T$, and the random variable associated with $p_t$ will be denoted by $X_t$.  Note that  $T\leq (n+1)^{|\mc{X}|}$.

We will be interested in four different $n$-copy decompositions of $\ket{\Psi}^{AB}$, where in all cases we assume that $n$ is being taken sufficiently large.  

\medskip

\noindent\textbf{Decomposition 1:}  The first decomposition is based on the coherence distillation protocol presented in Ref.~\cite{Winter-2015a}.  It involves forming good sets $S_m$ in the sense of Corollary \ref{cor:Cover} and w.r.t. the CQ channel $\mc{W}_{\textrm{CQ}}$.  For each typical type $p_t$, consider a partitioning of $T^n_{p_t}$ according to Corollary \ref{cor:Cover}.  Then for each each $x^n\in T^n_{[p]_\delta}$, we can relabel $x^n\to(t,m,s)$ where $p_t$ is the typical type for which $x^n\in T^n_{p_t}$; $m$ is the block number within $T^n_{p_t}$ for which $x^n$ belongs (with $m=0$ labeling the small remainder block); and $s$ is the order of $x^n$ in the $m^{th}$ block.  For a fixed $t$, the range of $m$ and $s$ is $m=1,\cdots ,M_t$ and $s=1,\cdots S_t$, where 
\begin{align}
M_t&=\lfloor |T^n_{p_t}|/S_t\rfloor, & S_t&=\lceil 2^{n(I(X_t:B)_{\mc{W}_{\mathrm{CQ}}}+\delta)}\rceil.
\end{align}
The bit rates of $S_t$ and $M_t$ satisfy
\begin{align}
\bigg|\frac{1}{n}\log S_t-\mathrm{E}(\Psi)\bigg|&\leq O(\tau(\delta))\\
\bigg|\frac{1}{n}\log M_t-[S(A)_{\Delta(\Psi)}-\mathrm{E}(\Psi)]\bigg|&\leq O(\tau(\delta),\frac{\log n}{n}).
\end{align}
The first line follows from Fannes' Inequality and the fact that $I(X:B)_\Psi=\mathrm{E}(\Psi)$, while the second can be seen from $\frac{1}{n}\log |T^n_{p_t}|\approx H(X)$ and $\frac{1}{n}\log S_t\approx I(X:B)_\Psi$.  Since all but a vanishing small fraction of $x^n$ belong to $T^n_{[X]_\delta}$, we can thus write 
\begin{align}
\ket{\Psi}^{\otimes n}&=\sum_{x^n}\sqrt{p^n(x^n)}\ket{x^n}\ket{\psi_{x^n}}\notag\\
&\overset{\epsilon}{\approx}\sum_{t=1}^T \sqrt{q(t)}\ket{t}^{A_1}\frac{1}{\sqrt{M_t}}\sum_{m=0}^{M_t}\ket{m}^{A_2}\frac{1}{\sqrt{S_t}}\sum_{s=1}^{S_t}\ket{s}^{A_3}\ket{\psi_{tms}}^{B}\label{Eq:Decomp1a}
\end{align}
where $q(t)$ is the probability of typical type class $T^n_{p_t}$ (conditioned on the event $x^n\in T^n_{[p]_\delta}$).  Note that for a sequence $x^n$ labeled by $(t,m,s)$ we have that $S(\Delta(\psi_{x^n}))=S(\Delta(\psi_{tms}))=\sum_{x\in\mc{X}}N(x|x^n\in T^n_{p_t})S(\Delta(\psi_x))$, with the RHS being independent of $s$ and $m$.  Since $S(Y|X)_{\Delta(\Psi)}=\sum_{x\in\mc{X}}p(x)S(\Delta(\psi_x))$, the following bound is obtained,
\begin{equation}
\bigg|\frac{1}{n}S(\Delta(\psi_{x^n}))-S(Y|X)_{\Delta(\Psi)}\bigg|\leq \delta \sum_{x\in\mc{X}}S(\Delta(\psi_x)),
\end{equation}
which again follows from $\delta$-typicality.  For each typical type $p_t$, we now further restrict the sum over $m$ to only those values for which $S_m$ are good sets.  By Corollary \ref{cor:Cover} there are $M_t(1-\epsilon)$ such sets.  For these values of $m$, we have that 
\begin{equation}
\label{Eq:TypeCover}
\frac{1}{S_t}\sum_{s=1}^{S_t}\psi_{tms}\overset{\epsilon}{\approx}\frac{1}{|T^n_{p_t}|}\sum_{x^n\in T^n_{p_t}}\psi_{x^n},
\end{equation}
with the RHS being independent of $m$.  As Uhlmann's Theorem states that $F(\rho_1,\rho_2)=\max|\ip{\varphi_1}{\varphi_2}|$, where the maximization is taken over all purifications of $\rho_1$ and $\rho_2$ respectively \cite{Uhlmann-1976a, Jozsa-1994a}, the previous equation implies for each pair $(t,m)$ the existence of a unitary $U_{tm}$ acting on $A_3$ such that
\begin{align}
\frac{1}{\sqrt{S_t}}\sum_{s =1 }^{S_t}\ket{s}^{A_3}\ket{\psi_{tms}}^B&\overset{O(\epsilon)}{\approx}\frac{1}{\sqrt{S_t}}\sum_{s=1}^{S_t}(U_{tm}\ket{s}^{A_3})\ket{\psi_{tm_0s}}^{B}\label{Eq:Uhlmann1},
\end{align}
where $m_0\in \{1,\cdots ,M_t(1-\epsilon)\}$ is some fixed number.  We thus continue Eq.~\eqref{Eq:Decomp1a} by restricting the sum over $m$ to just good values and replacing the sum over $s$ with Eq.~\eqref{Eq:Uhlmann1}:
\begin{align}
\ket{\Psi}^{\otimes n}&\overset{O(\epsilon)}{\approx}\sum_{t=1}^T \sqrt{q(t)}\ket{t}^{A_1}\frac{1}{\sqrt{M_t(1-\epsilon)}}\sum_{m=1}^{M_t(1-\epsilon)}\ket{m}^{A_2}\frac{1}{\sqrt{S_t}}\sum_{s=1}^{S_t}\left(U_{tm}\ket{s}^{A_3}\right)\ket{\psi_{t m_0s}}^B\label{Eq:decomp1}.
\end{align}

\medskip

\noindent\textbf{Decompositions 2 and 3:}  The second and third decompositions are built from $(n,\epsilon)$ codes for the channels $\mc{W}_{\textrm{CC}}$ and $\mc{W}_{\textrm{CQ}}$ respectively.  The structure of the decompositions is based on the entanglement-assisted and GHZ distillation schemes of Ref.~\cite{Smolin-2005a}.  

First we turn to $\mc{W}_{\textrm{CC}}$.  For every typical type $p_t$, consider a partitioning of $T^n_{p_t}$ according to Corollary \ref{cor:HSWcover}.  Then each $x^n\in T^n_{[p]_\delta}$ can be relabeled $x^n\to(t,l,c)$ where $p_t$ is the typical type for which $x^n\in T^n_{p_t}$; $l$ is the code $C_l$ within $T^n_{p_t}$ for which $x^n$ belongs; and $c$ is the order of $x^n$ in the $l^{th}$ code.  For a fixed $t$, the range of $l$ and $c$ is $l=1,\cdots ,L_t$ and $c=1,\cdots C_t$, where 
\begin{align}
L_t&=\lfloor 2^{n(H(X_t)-I(X_t:Y)_{\mc{W}_{\rm CC}}+2\delta)}\rfloor, & C_t&=\lceil 2^{n(I(X_t:Y)_{\mc{W}_{\rm CC}}-\delta)}\rceil.
\end{align}
The bit rates of $L_t$ and $C_t$ satisfy
\begin{align}
\bigg|\frac{1}{n}\log L_t-[H(X)-I(X:Y)_{\Delta(\Psi)}]\bigg|&\leq O(\tau(\delta),\frac{\log n}{n})\\
\bigg|\frac{1}{n}\log C_t-I(X:Y)_{\Delta(\Psi)}\bigg|&\leq O(\tau(\delta)).
\end{align}
By discarding non-typical sequences and the fraction $3\epsilon$ of $x^n$ not belonging to an $(n,\epsilon)$ code, we obtain the approximation
\begin{align}
\ket{\Psi}^{\otimes n}
&\overset{O(\epsilon)}{\approx}\sum_{t=1}^T \sqrt{q(t)}\ket{t}^{A_1}\frac{1}{\sqrt{L_t(1-3\epsilon)}}\sum_{l=1}^{L_t(1-3\epsilon)}\ket{l}^{A_2}\frac{1}{\sqrt{C_t}}\sum_{c=1}^{C_t}\ket{c}^{A_3}\ket{\psi_{tlc}}^B.
\end{align}
For every $(t,l)$ define the state $\ket{\chi_{tl}}=\frac{1}{\sqrt{C_t}}\sum_{c=1}^{C_t}\ket{c}^{A_3}\ket{\psi_{tlc}}^B$.  By $\epsilon$-decodability of the channel $\mc{W}_{\rm CC}$ there exists a family of decoding POVMs $(D^{(tl)}_c)_{c=1}^{C_t}$ such that
\[\frac{1}{C_t}\sum_{c=1}^{C_t}\tr[\Delta(\psi_{tlc})D^{(tl)}_c]>1-\epsilon.\]
Note that $\tr[\Delta(\psi_{tlc})D^{(tl)}_c]=\tr[\Delta(\psi_{tlc})\Delta(D^{(tl)}_c)]$, and so without loss of generality, we can assume that the $D_c^{(tl)}$ are diagonal in the incoherent basis.  We consider a dilation of the $(t,l)^{th}$ POVM.  To do so, introduce the isometries $Y:B\to B B_1$ and $W_{tl}:B\to B B_2$ as
\begin{align}
Y&=\sum_{y^n\in\mc{Y}^n}\ket{y^n}^{B}\bra{y^n}^{B}\otimes \ket{y^n}^{B_1}\notag\\
W_{tl}&=\sum_{c=1}^{C_t}\sqrt{D_c^{(tl)}}\otimes\ket{c}^{B_2}.
\label{Eq:POVMisometry}
\end{align}
Crucially, both $Y$ and $W_{tl}$ represent incoherent operations.  Define the state
\begin{align}
\ket{\widehat{\chi}_{tl}}&=(\mbb{I}^{A_3}\otimes W_{tl}Y)\ket{\chi_{tl}}\notag\\
&=\frac{1}{\sqrt{C_t}}\sum_{c=1}^{C_t}\ket{c}^{A_3}\sum_{y^n\in\mc{Y}^n}\sum_{c'=1}^{C_t}\ip{y^n}{\psi_{tlc}}\sqrt{D^{(tl)}_{c'}}\ket{y^n}^{B}\ket{y^n}^{B_1}\ket{c'}^{B_2}.
\end{align}
We want to show that this state is $\epsilon$-close to the state
\begin{align}
\ket{\widehat{\widehat{\chi}}_{tl}}&=\frac{1}{\sqrt{C_t}}\sum_{c=1}^{C_t}\ket{c}^{A_3}\sum_{y^n\in\mc{Y}^n}\ip{y^n}{\psi_{tlc}}\ket{y^n}^{B}\ket{y^n}^{B_1}\ket{c}^{B_2},
\end{align}
which would imply that Bob coherently decode $\ket{c}$ from $\ket{\chi_{tl}}$ without  disturbing the state that much.  To this end, first note that $\ket{\widehat{\chi}_{tl}}\overset{O(\epsilon)}{\approx}(\mbb{I}^{A_3}\otimes\sqrt{X_{tl}})\ket{\widehat{\widehat{\chi}}_{tl}}$ where
\[X_{tl}=\sum_{c=1}^{C_t} D_c^{(tl)}\otimes\mbb{I}^{B_1}\otimes\op{c}{c}^{B_2}.\]
The approximation $\overset{O(\epsilon)}{\approx}$ here can be seen from the fact that
\[\bra{\widehat{\chi}_{tl}}(\mbb{I}^{A_3}\otimes\sqrt{X_{tl}})\ket{\widehat{\widehat{\chi}}_{tl}}=\frac{1}{C_t}\sum_{c=1}^{C_t}\tr[\Delta(\psi_{tlc})D^{(tl)}_c]>1-\epsilon.\]
Then applying Lemma \ref{Lem:gentle} to $\tr[(\mbb{I}^{A_3}\otimes X_{tl})\widehat{\widehat{\chi}}_{tl}]>1-\epsilon$, we can conclude that $(\mbb{I}^{A_3}\otimes\sqrt{X_{tl}})\ket{\widehat{\widehat{\chi}}_{tl}}\overset{O(\epsilon)}{\approx}\ket{\widehat{\widehat{\chi}}_{tl}}$.  Therefore, $\ket{\widehat{\chi}_{tl}}\overset{O(\epsilon)}{\approx}\ket{\widehat{\widehat{\chi}}_{tl}}$ and so
\begin{align}
\ket{\Psi}^{\otimes n}&\overset{O(\epsilon)}{\approx}\sum_{t=1}^T \sqrt{q(t)}\ket{t}^{A_1}\frac{1}{\sqrt{L_t(1-3\epsilon)}}\sum_{l=1}^{L_t(1-3\epsilon)}\ket{l}^{A_2}\frac{1}{\sqrt{C_t}}\sum_{c=1}^{C_t}\ket{c}^{A_3}\sum_{y^n\in\mc{Y}^n}\ip{y^n}{\psi_{tlc}}W_{tl}^\dagger \left(\ket{y^n}^{B}\ket{c}^{B_2}\right)\notag\\
&=\sum_{t=1}^T \sqrt{q(t)}\ket{t}^{A_1}\frac{1}{\sqrt{L_t(1-3\epsilon)}}\sum_{l=1}^{L_t(1-3\epsilon)}\ket{l}^{A_2}\frac{1}{\sqrt{C_t}}\sum_{c=1}^{C_t}\ket{c}^{A_3}W_{tl}^\dagger \left(\ket{\psi_{tlc}}^{B}\ket{c}^{B_2}\right)\notag\\
&=\sum_{t=1}^T \sqrt{q(t)}\ket{t}^{A_1}\frac{1}{\sqrt{L_t(1-3\epsilon)}}\sum_{l=1}^{L_t(1-3\epsilon)}\ket{l}^{A_2}\frac{1}{\sqrt{C_t}}\sum_{c=1}^{C_t}\ket{c}^{A_3}W_{tl}^\dagger \left(\Pi_{tlc}\ket{\psi_{tl_0c_0}}^{B}\ket{c}^{B_2}\right),
\label{Eq:decomp3}
\end{align}
where $\Pi_{tlc}$ permutes $\ket{\psi_{tl_0c_0}}$ into $\ket{\psi_{tlc}}$, for some fixed $l_0\in 1,\cdots, L_t$ and $c_0\in 1,\cdots, C_t$.  Recall that for each type $t$, each $\ket{\psi_{tlc}}$ is a sequence $\ket{\psi_{tlc}}=\ket{\psi_{x_1}}\ket{\psi_{x_2}}\cdots\ket{\psi_{x_n}}$ related to one another through a permutation of the $\ket{\psi_{x_i}}$.

We now repeat an analogous decomposition for the CQ channel $\mc{W}_{\textrm{CQ}}$.  Since this will involve a different covering of the type classes we use a different labeling $x^n\to(t,\overline{l},\overline{c})$.  By the same arguments as above, the decomposition takes the form
\begin{align}
\ket{\Psi}^{\otimes n}&\overset{O(\epsilon)}{\approx}\sum_{t=1}^T \sqrt{q(t)}\ket{t}^{A_1}\frac{1}{\sqrt{\ovl{L_t}(1-3\epsilon)}}\sum_{\ovl{l}=1}^{\ovl{L_t}(1-3\epsilon)}\ket{\ovl{l}}^{A_2}\frac{1}{\sqrt{\ovl{C_t}}}\sum_{\ovl{c}=1}^{\ovl{C_t}}\ket{\ovl{c}}^{A_3}\ovl{W}_{t\ovl{l}}^\dagger \left(\Pi_{t\ovl{l}\ovl{c}}\ket{\psi_{t\ovl{l_0}\ovl{c_0}}}^{B}\ket{\ovl{c}}^{B_2}\right).
\label{Eq:decomp4}
\end{align}
Here, like before, $\ovl{W}_{tl}$ is an isometry for the $(t,l)^{th}$ decoding POVM of $\mc{W}_{\textrm{CQ}}$ as in Eq.~\eqref{Eq:POVMisometry}.  However $\ovl{W}_{tl}$ will in general not be incoherent.

\medskip

\noindent\textbf{Decomposition 4:}  The fourth decomposition is a hybrid of decompositions 1 and 2.  It begins with Eq.~\eqref{Eq:Decomp1a} and the fact that the sum over $m$ includes $M_t(1-\epsilon)$ good sets in the sense that
\begin{equation}
\frac{1}{S_t}\sum_{s=1}^{S_t}\psi_{tms}\overset{\epsilon}{\approx}\frac{1}{|T^n_{p_t}|}\sum_{x^n\in T^n_{p_t}}\psi_{x^n}
\end{equation}
for these good values of $m$.  In this decomposition, we now replace the RHS by $\mc{W}_{\rm CC}$ channel codes.  That is, we use Corollary \ref{cor:HSWcover} to write
\[\frac{1}{S_t}\sum_{s=1}^{S_t}\psi_{tms}\overset{\epsilon}{\approx}\frac{1}{|T^n_{p_t}|}\sum_{x^n\in T^n_{p_t}}\psi_{x^n}=\frac{1}{L_tC_t}\sum_{l=1}^{L_t}\sum_{c=1}^{C_t} \psi_{tlc}.\]
Uhlmann's Theorem again implies that for each good value of $m$ there exists a right orthogonal matrix $V_{tm}:A_3A_4\to A_3$ (with $V_{tm}V_{tm}^\dagger=\mbb{I}^{A_3}$)
 such that 
\begin{align}
\frac{1}{\sqrt{S_t}}\sum_{s =1 }^{S_t}\ket{s}^{A_3}\ket{\psi_{tms}}^B&\overset{O(\epsilon)}{\approx}\frac{1}{\sqrt{L_tT_t}}\sum_{l=1}^{L_t}\sum_{c=1}^{C_t} \left(V_{tm}\ket{lc}^{A_3A_4}\right)\psi_{tlc}^B.\label{Eq:Uhlmann2}
\end{align}
Hence by restricting to good values of $m$ and $(n,\epsilon)$ channel codes, we have the decomposition
\begin{align}
\ket{\Psi}^{\otimes n}&\overset{O(\epsilon)}{\approx}\sum_{t=1}^T \sqrt{q(t)}\ket{t}^{A_1}\frac{1}{\sqrt{M_t(1-\epsilon)}}\sum_{m=1}^{M_t(1-\epsilon)}\ket{m}^{A_2}
\frac{1}{\sqrt{L_tC_t(1-3\epsilon)}}\sum_{l=1}^{L_t(1-3\epsilon)}\sum_{c=1}^{C_t}\left(V_{tm}\ket{lc}^{A_3A_4}\right)\ket{\psi_{tlc}}^B.
\end{align}
Finally, similar to the construction in decomposition 2, decoding isometries $W_{tl}$ exist for Bob so that the state can be expressed as
\begin{align}
\ket{\Psi}^{\otimes n}&\overset{O(\epsilon)}{\approx}\sum_{t=1}^T \sqrt{q(t)}\ket{t}^{A_1}\frac{1}{\sqrt{M_t(1-\epsilon)}}\sum_{m=1}^{M_t(1-\epsilon)}\ket{m}^{A_2}
\frac{1}{\sqrt{L_tC_t(1-3\epsilon)}}\sum_{l=1}^{L_t(1-3\epsilon)}\sum_{c=1}^{C_t}\left(V_{tm}\ket{lc}^{A_3A_4}\right)W_{tl}^\dagger \left(\Pi_{tlc}\ket{\psi_{tl_0c_0}}^{B}\ket{c}^{B_2}\right),\label{Eq:decomp2}
\end{align}
for some fixed $(l_0,c_0)$.

\medskip

\noindent\textbf{Summary of Code Construction:}

The four decompositions we will use are given by Eqns. \eqref{Eq:decomp1}, \eqref{Eq:decomp3}, \eqref{Eq:decomp4}, and \eqref{Eq:decomp2}.  For convenience, we recall them here:
\begin{subequations}
\begin{align}
\ket{\Psi}^{\otimes n}&\overset{O(\epsilon)}{\approx}\sum_{t=1}^T \sqrt{q(t)}\ket{t}^{A_1}\frac{1}{\sqrt{M_t(1-\epsilon)}}\sum_{m=1}^{M_t(1-\epsilon)}\ket{m}^{A_2}\frac{1}{\sqrt{S_t}}\sum_{s=1}^{S_t}\left(U_{tm}\ket{s}^{A_3}\right)\ket{\psi_{t m_0s}}^B,\label{Eq:decomp1a}\\
\ket{\Psi}^{\otimes n}&\overset{O(\epsilon)}{\approx}\sum_{t=1}^T \sqrt{q(t)}\ket{t}^{A_1}\frac{1}{\sqrt{L_tC_t(1-3\epsilon)}}\sum_{l=1}^{L_t(1-3\epsilon)}\sum_{c=1}^{C_t}\ket{l}^{A_2}\ket{c}^{A_3}W_{tl}^\dagger \left(\Pi_{tlc}\ket{\psi_{tl_0c_0}}^{B}\ket{c}^{B_2}\right),\label{Eq:decomp3c}\\
\ket{\Psi}^{\otimes n}&\overset{O(\epsilon)}{\approx}\sum_{t=1}^T \sqrt{q(t)}\ket{t}^{A_1}\frac{1}{\sqrt{\ovl{L_t}\ovl{C_t}(1-3\epsilon)}}\sum_{\ovl{l}=1}^{\ovl{L_t}(1-3\epsilon)}\sum_{\ovl{c}=1}^{\ovl{C_t}}\ket{\ovl{l}}^{A_2}\ket{\ovl{c}}^{A_3}\ovl{W}_{t\ovl{l}}^\dagger \left(\Pi_{t\ovl{l}\ovl{c}}\ket{\psi_{t\ovl{l_0}\ovl{c_0}}}^{B}\ket{\ovl{c}}^{B_2}\right),\label{Eq:decomp4d}\\
\ket{\Psi}^{\otimes n}&\overset{O(\epsilon)}{\approx}\sum_{t=1}^T \sqrt{q(t)}\ket{t}^{A_1}\frac{1}{\sqrt{M_t(1-\epsilon)}}\sum_{m=1}^{M_t(1-\epsilon)}\ket{m}^{A_2}
\frac{1}{\sqrt{L_tC_t(1-3\epsilon)}}\sum_{l=1}^{L_t(1-3\epsilon)}\sum_{c=1}^{C_t}\left(V_{tm}\ket{lc}^{A_3A_4}\right)W_{tl}^\dagger \left(\Pi_{tlc}\ket{\psi_{tl_0c_0}}^{B}\ket{c}^{B_2}\right),\label{Eq:decomp2b}
\end{align}
\end{subequations}
with bit rates
\begin{align}
\frac{1}{n}\log T& \leq O(\tfrac{\log n}{n}),\label{Eq:rate-T}\\
\bigg|\frac{1}{n}\log S_t-\mathrm{E}(\Psi)\bigg|&\leq O(\tau(\delta)),\label{Eq:rate-St}\\
\bigg|\frac{1}{n}\log M_t-[S(X)_{\Delta(\Psi)}-\mathrm{E}(\Psi)]\bigg|&\leq O(\tau(\delta),\frac{\log n}{n}),\label{Eq:rate-Mt}\\
\bigg|\frac{1}{n}\log L_t-S(X|Y)_{\Delta(\Psi)}\bigg|&\leq O(\tau(\delta),\frac{\log n}{n}),\label{Eq:rate-Lt}\\
\bigg|\frac{1}{n}\log C_t-I(X:Y)_{\Delta(\Psi)}\bigg|&\leq O(\tau(\delta)),\label{Eq:rate-Ct}\\
\bigg|\frac{1}{n}S(\Delta(\psi_{x^n}))-S(Y|X)_{\Delta(\Psi)}\bigg|&\leq O(\delta),\label{Eq:rate-Bob-psi}\\
\bigg|\frac{1}{n}\log \overline{L_t}-[S(X)_{\Delta(\Psi)}-\mathrm{E}(\Psi)]\bigg|&\leq O(\tau(\delta),\frac{\log n}{n}),\label{Eq:rate-qLt}\\
\bigg|\frac{1}{n}\log \overline{C_t}-\mathrm{E}(\Psi)\bigg|&\leq O(\tau(\delta))\label{Eq:rate-qCt}.
\end{align}

\section{Proofs of Main Text Theorems/Lemmas and Expanded Discussion}

\subsection{Proof of Theorem 1}

\begin{thm1}
\label{thm1}
For a pure state $\ket{\Psi}^{AB}$ the following triples are achievable coherence-entanglement formation rates
\begin{align}
(R_A,R_B,E^{co})&=\left(\;0,\;S(Y|X)_{\Delta(\Psi)}\;,\; S(X)_{\Delta(\Psi)}\right)\label{Eq:CostTriple1}\\
(R_A,R_B,E^{co})&=\left(S(X)_{\Delta(\Psi)},\;S(Y|X)_{\Delta(\Psi)},\; \mathrm{E}(\Psi)\right)\label{Eq:CostTriple2}\\
(R_A,R_B,E^{co})&=\left(0,\;0,\; S(XY)_{\Delta(\Psi)}\right)\label{Eq:CostTriple3}
\end{align}
as well as the points obtained by interchanging $X\leftrightarrow Y$ in Eqns. \eqref{Eq:CostTriple1} -- \eqref{Eq:CostTriple3}.  Moreover, these points are optimal in the sense that any achievable rate triple must satisfy
\begin{align}
E^{co}&\geq \mathrm{E}(\Psi)&R_A+R_B&\geq S(XY)_{\Delta(\Psi)}&R_B+E^{co}&\geq S(XY)_{\Delta(\Psi)}.
\label{Eq:CostLB}
\end{align}
\end{thm1}
\begin{example}
Before getting to the proof, we give an example of the resource trade-off provided by Theorem \ref{thm1}.  Consider the formation of $\ket{\Psi}^{AB}=\sqrt{\lambda}\ket{+}\ket{0}+\sqrt{1-\lambda}\ket{-}\ket{1}$ for $0<\lambda<1/2$.  This state has entanglement $\mathrm{E}(\Psi)=S(A)_\Psi<1$.  The lower bounds of Theorem \ref{thm1} say that any formation protocol that has $E^{co}=\mathrm{E}(\Psi)$ must have coherence rates satisfying $R_B\geq S(AB)_{\Delta(\Psi)}-\mathrm{E}(\Psi)=1$.  However, if Alice and Bob would rather use more eCoBits than CoBits, they can reduce Bob's local coherence rate.  Namely, rate \eqref{Eq:CostTriple1} gives $R_A=0$ and $R_B=S(Y|X)_{\Delta(\Psi)}<1$.
\end{example}

\begin{proof}
The lower bounds of Eq. \eqref{Eq:CostLB} follow from the coherence cost rates using global operations (Theorem 3 of \cite{Winter-2015a} as well as \cite{Yuan-2015a} and Theorem \ref{thm5} below), and the fact that one eCoBit can be converted into one CoBit using LIOCC.  Indeed, $\ket{\Phi_{AB}}\to\ket{\Phi_B}$ when Alice performs the incoherent measurement with Kraus operators $\{\op{0}{+},\op{1}{-}\}$, and then Bob performs $\sigma_Z$ iff Alice obtains outcome $\ket{1}$.

Moving to achievability, we first prove Eq.~\eqref{Eq:CostTriple1}.  The protocol is based on decomposition \eqref{Eq:decomp3c}.  Alice and Bob share $\log E_n$ eCoBits expressed as
\begin{equation}
\label{Eq:EntanglementInitial}
\frac{1}{\sqrt{E_n}}\sum_{t=1}^T\sum_{l=1}^{L_t(1-3\epsilon)}\sum_{c=1}^{C_t}\ket{tt}^{A_1B_0}\ket{ll}^{A_2B_1}\ket{cc}^{A_3B_2},
\end{equation}
where $E_n=\prod_{t=1}^TL_t(1-3\epsilon)C_t$, while Bob has an additional $n[ S(Y|X)_{\Delta(\Psi)}+O(\delta)]$ CoBits.  In the first step, Alice and Bob deterministically transform their state into
\[\sum_{t=1}^T\sqrt{q(t)}\frac{1}{\sqrt{L_t(1-3\epsilon)}}\sum_{l=1}^{L_t(1-3\epsilon)}\frac{1}{\sqrt{C_t}}\sum_{c=1}^{C_t}\ket{tt}^{A_1B_0}\ket{ll}^{A_2B_1}\ket{cc}^{A_3B_2}\]
where $q(t)$ is given in Eq. \eqref{Eq:decomp3c}.  This can always be done via the majorization criterion of Lemma \ref{lem:majorization} below.  Using his local CoBits, Bob first prepares $\ket{\psi_{tl_0c_0}}^B$ on ancillary system $B$, which can be done to arbitrary precision for sufficiently large $n$ (Theorem 3 of \cite{Winter-2015a} as well as \cite{Yuan-2015a} and Remark \ref{remark:unitary2} below).  Next he performs the incoherent unitary $W^\dagger_{tl}\Pi_{tlc}$ conditioned on $\ket{t}^{B_0}\ket{l}^{B_1}\ket{c}^{B_2}$.  More precisely, $\Pi_{tlc}$ is performed on system $B$ conditioned on $B_0B_1B_2$, and $W^{\dagger}_{tl}$ is performed on $BB_2$ conditioned on $B_0B_1$.  Finally, he decouples his registers $B_0B_1$.  To accomplish this, he performs a generalized incoherent measurement $\{K_{t,l}=\op{00}{\gamma_{t,l}}^{B_0B_1}\}_{t,l=1}^{T,L_t(1-3\epsilon)}$ where $\ket{\gamma_{t,l}}=\sum_{t'=1}^{T}\sum_{l'=1}^{L_t(1-3\epsilon)}e^{i 2\pi \left(\tfrac{(t-1) (t'-1)}{T}+\tfrac{(l-1) (l'-1)}{L_t(1-3\epsilon)}\right)}\ket{t'}\ket{l'}$.  For outcome $K_{t,l}$, Bob announces the result to Alice, and she performs the incoherent unitary
\begin{equation}
\label{Eq:AliceErrorCorrect}
U_{tl}=\sum_{t'=1}^{T}\sum_{l'=1}^{L_t(1-3\epsilon)}e^{-i 2\pi \left(\tfrac{(t-1) (t'-1)}{T}+\tfrac{(l-1) (l'-1)}{L_t(1-3\epsilon)}\right)}\op{t}{t'}^{A_1}\op{l}{l'}^{A_2}.
\end{equation}
The desired state is thus obtained: 
\[\sum_{t=1}^T \sqrt{q(t)}\ket{t}^{A_1}\frac{1}{\sqrt{L_t(1-3\epsilon)}}\sum_{l=1}^{L_t(1-3\epsilon)}\ket{l}^{A_2}\frac{1}{\sqrt{C_t}}\sum_{c=1}^{C_t}\ket{c}^{A_3}W_{tl}^\dagger \left(\Pi_{tlc}\ket{\psi_{tl_0c_0}}^{B}\ket{c}^{B_2}\right).\]
Asymptotically, the consumption rates of entanglement and coherence approach Eq. \eqref{Eq:CostTriple1}.

Now we prove achievability of Eq.~\eqref{Eq:CostTriple2}.  The protocol is based on decomposition \eqref{Eq:decomp1a}.  Alice and Bob share $\log E'_n$ eCoBits expressed as
\begin{equation}
\frac{1}{\sqrt{E'_n}}\sum_{t=1}^T\sum_{s=1}^{S_t}\ket{tt}^{A_1B_1}\ket{ss}^{A_3B_3},
\end{equation}
where $E'_n=\prod_{t=1}^TS_t$, while Alice has an additional $n[ S(X)_{\Delta(\Psi)}+O(\tau(\delta),\frac{\log n}{n})]$ CoBits and Bob has $n[ S(Y|X)_{\Delta(\Psi)}+O(\delta) ]$ CoBits.  In the first step of the protocol, Alice and Bob again deterministically transform their entanglement into 
\[\sum_{t=1}^T\sqrt{q(t)}\ket{tt}^{A_1B_1}\frac{1}{\sqrt{S_t}}\sum_{s=1}^{S_t}\ket{ss}^{A_3B_3}\]
using LIOCC.  Using local CoBits, Bob prepares $\ket{\psi}^{B}_{tm_{0}s}$, conditioned on $\ket{ts}^{B_1B_3}$.  After this, he decouples his registers $B_1B_3$ using a measurement described above, and Alice performs a suitable incoherent unitary similar to Eq.~\eqref{Eq:AliceErrorCorrect}.  At this point, Alice and Bob share
\[\sum_{t=1}^T\sqrt{q(t)}\ket{t}^{A_1}\frac{1}{\sqrt{S_t}}\sum_{s=1}^{S_t}\ket{s}^{A_3}\ket{\psi_{tm_{0}s}}^{B}.\]
Next, Alice splits her coherence into two parts $\ket{\kappa_1}\ket{\kappa_2}$, where $\ket{\kappa_i}=\frac{1}{\sqrt{\kappa_i}}\sum_{x=1}^{\kappa_i}\ket{x}$ for $\kappa_1=n[S(X)_{\Delta(\Psi)}-\mathrm{E}(\Psi)+O(\tau(\delta),\frac{\log n}{n})]$ and $\kappa_2=n[\mathrm{E}(\Psi)+O(\tau(\delta),\frac{\log n}{n})]$.  Using $\ket{\kappa_1}$, she implements a rotation $\ket{0}\to \frac{1}{\sqrt{M_t(1-\epsilon)}}\sum_{m=1}^{M_t}\ket{m}$, conditioned on $\ket{t}$.  Using $\ket{\kappa_2}$, she implements another unitary rotation $\ket{s}\to U_{tm}\ket{s}$ conditioned on $\ket{t}\ket{m}$, for $s=1,\cdots, S_t$.  The desired state is thus obtained: 
\[\sum_{t=1}^T \sqrt{q(t)}\ket{t}^{A_1}\frac{1}{\sqrt{M_t(1-\epsilon)}}\sum_{m=1}^{M_t(1-\epsilon)}\ket{m}^{A_2}\frac{1}{\sqrt{S_t}}\sum_{s=1}^{S_t}\left(U_{tm}\ket{s}^{A_3}\right)\ket{\psi_{t m_0s}}^B.\]
Asymptotically, the consumption rates of entanglement and coherence approach Eq. \eqref{Eq:CostTriple2}.

Finally, the achievability of Eq.~\eqref{Eq:CostTriple3} follows from Eq.~\eqref{Eq:CostTriple1} and the fact that every eCoBit can be deterministically transformed into a CoBit for Bob using LIOCC.

\end{proof}

\subsection{Proof of Lemmas 2 and 3}

\begin{lem1}
\label{lem2}
An arbitrary $d\times d$ unitary operator $U$ can be performed on a system using incoherent operations and $\lceil \log d\rceil$ CoBits.
\end{lem1}

\begin{proof}
Let us introduce an orthonormal basis $\{\ket{b_{jk}}\}_{j,k=0}^d$ for a $d\otimes d$ bipartite system $SS'$ consisting of maximally entangled states 
\begin{align}
\ket{b_{jk}}^{SS'}=\mbb{I}\otimes W_{jk}\ket{\Phi^{(d)}_{SS'}}
\end{align} 
where 
\[\ket{\Phi^{(d)}_{SS'}}=\frac{1}{\sqrt{d}}\sum_{i=0}^{d-1}\ket{i}^{S}\otimes\ket{i}^{S'},\qquad W_{jk}=\sum_{l=0}^{d-1}\tau^{lj}\op{k+l}{l},\]
with $\tau=e^{2\pi i/d}$ and addition is modulo $d$.  The $\{\ket{b_{jk}}\}$ generalize the Bell basis in higher dimensions.  Note that each unitary $W_{jk}$ is an incoherent operation.

Suppose now that we wish to perform an arbitrary $d\otimes d$ unitary $U$ on some state $\ket{\psi}$.  We can accomplish this incoherently using $\lceil \log d\rceil$ CoBits through the following procedure.  First, let $\ket{\psi}$ belong to system $S$, and express $\ket{\psi}^S=[\psi]\ket{\Phi^{(d)}_S}$, where $[\psi]$ is a $d\times d$ complex matrix and $\ket{\Phi^{(d)}_S}=1/\sqrt{d}\sum_{i=0}^{d-1}\ket{i}^S$.  We next introduce $\lceil \log d\rceil$ CoBits $\ket{\Phi_{S'}}^{\otimes \lceil \log d\rceil}$ on system $S'$.  This is deterministically transformed into  $\ket{\Phi^{(d)}_{S'}}=1/\sqrt{d}\sum_{i=0}^{d-1}\ket{i}^{S'}$, which can always be accomplished since $\ket{\Phi^{(d)}_{S'}}$ majorizes $\ket{\Phi_{S'}}^{\otimes \lceil \log d\rceil}$ (see Theorem 1 of \cite{Winter-2015a} as well as Ref. \cite{Chitambar-2016a}).  An additional system $S''$ is then introduced in state $\ket{0}^{S''}$ and an entangling incoherent operation is performed to obtain
\[
\ket{\Phi_{S'}^{(d)}}\ket{0}\to\ket{\Phi^{(d)}_{S'S''}}.
\]
Thus at this point the state across all three systems 
\begin{equation}
\label{Eq:PreMeasureState1}
\ket{\psi}^S\ket{\Phi^{(d)}_{S'S''}}=\frac{1}{d}\sum_{i,j=0}^{d-1}\left([\psi]\ket{i}^S\right)\otimes \ket{j}^{S'}\otimes\ket{j}^{S''}.
\end{equation}
An incoherent measurement $\{M_{jk}\}_{j,k=0}^{d-1}$ is then performed on systems $SS'$ given by
\begin{equation}
\label{Eq:Kraus-Unitary-Implementation}
M_{jk}=\ket{jk}\left(\bra{b_{jk}}U^{S}\otimes \mbb{I}^{S'} \right)=\op{jk}{\Phi_{SS'}^{(d)}}\mbb{I}\otimes U^T W^\dagger_{jk}.
\end{equation}
From Eq.~\eqref{Eq:PreMeasureState1}, we see that outcome $jk$ generates the (unnormalized) post-measurement state
\begin{align}
\frac{1}{d}\sum_{m,n=0}^{d-1}&\op{jk}{\Phi_{SS'}^{(d)}}[\psi] \otimes U^T W^\dagger_{jk}\ket{mn}^{SS'}\otimes\ket{n}^{S''}\notag\\
&=\frac{1}{d}\sum_{m,n=0}^{d-1}\op{jk}{\Phi_{SS'}^{(d)}} \mbb{I}\otimes [\psi]^TU^TW^\dagger_{jk}\ket{mn}^{SS'}\otimes\ket{n}^{S''}\notag\\
&=\frac{1}{d}\ket{jk}\otimes \frac{1}{\sqrt{d}}\sum_{m=0}^{d-1}W_{jk}^*U[\psi]\ket{m}^{S''}\notag\\
&=\frac{1}{d}\ket{jk}\otimes W_{jk}^* U\ket{\psi}^{S''}.
\end{align}
Therefore, after applying the (incoherent) rotation $W_{jk}^*$ on system $S''$, the state $U\ket{\psi}$ is obtained.
\end{proof}

\begin{remark}
\label{remark:unitary1}
Lemma \ref{lem2} can easily be extended to performing controlled unitaries of the form $\sum_{r=1}^R \op{r}{r}^C\otimes U_r$, where $C$ is the control system of dimension $R$.  When each $U_r$ acts on a $d$-dimensional system, the amount of CoBits needed to perform this operation is given by $\lceil \log d\rceil$.  Indeed, the above protocol is repeated with a replacement in Eq. \eqref{Eq:Kraus-Unitary-Implementation} of an incoherent measurement $\{M_{jk}\}_{j,k=0}^{d-1}$ performed on systems $CSS'$ given by
\begin{equation}
M_{jk}=\sum_{r=1}^R \op{r}{r}^C\otimes \ket{jk}\left(\bra{b_{jk}}U_r^{S}\otimes \mbb{I}^{S'} \right).
\end{equation}
\end{remark}
\begin{remark}
\label{remark:unitary2}
Lemma \ref{lem2} offers an alternative protocol for coherence dilution in pure states \cite{Winter-2015a}.  For an arbitrary pure state $\ket{\psi}=\sum_{x=1}^d\sqrt{p(x)}e^{i\theta_x}\ket{x}$, we can transform $\ket{+}^{n}\to\overset{\epsilon}{\approx}\ket{\psi}^{\otimes \lfloor nR\rfloor}$ at any rate $R<S(X)_{\Delta(\psi)}$ as $n\to\infty$.  To see this we consider the $n$-copy decomposition of $\ket{\psi}$ into its typical and atypical parts \cite{Winter-2015a}:
\begin{equation}
\ket{\psi}^{\otimes n}=\sqrt{p^n(T^n_{[p]_\delta})}\ket{\text{typical}}+\sqrt{1-p^n(T^n_{[p]_\delta})}\ket{\text{atypical}}.
\end{equation}
Since $|T^n_{[p]_\delta}|\leq 2^{n[S(X)_{\Delta(\psi)}+\delta]}$ and $p^n(T^n_{[p]_\delta})\to 1$ for arbitrary $\delta>0$ and as $n\to\infty$, dilution is achieved by performing a unitary operation that rotates $\ket{+}^n$ to $\ket{\text{typical}}$.  Since $\ket{\text{typical}}$ is an element in a $|T^n_{[p]_\delta}|$-dimensional space, Lemma \ref{lem2} implies that such a unitary can be implemented by incoherent operations at a coherence consumption rate arbitrary close to $S(X)_{\Delta(\psi)}$.
\end{remark}

\begin{lem1}
\label{lem:majorization}
Let $\ket{\psi}^{AB}$ and $\ket{\phi}^{AB}$ be two bipartite pure states with squared Schmidt coefficients being $\vec{\tau}(\psi)$ and $\vec{\tau}(\phi)$ respectively.  Suppose that Alice and Bob's incoherent bases are Schmidt bases for both $\ket{\psi}^{AB}$ and $\ket{\phi}^{AB}$, and suppose that $\vec{\tau}(\phi)$ majorizes $\vec{\tau}(\psi)$ (i.e. $\vec{\tau}(\psi)\prec \vec{\tau}(\phi)$).  Then there exists an LIOCC protocol transforming $\ket{\psi}^{AB}\to \ket{\phi}^{AB}$ with probability one.
\end{lem1}

\begin{proof}
Recall that a probability distribution $\vec{y}=(y_1,\cdots,y_n)$ majorizes another distribution $\vec{x}=(x_1,\cdots,x_n)$ if $\sum_{j=1}^k y_j^\downarrow\geq\sum_{j=1}^k x_j^\downarrow$ for all $k=1,\cdots,n$, where $y_j^\downarrow$ is the components of $\vec{y}$ in non-increasing order and likewise for $x_j^\downarrow$.  Without loss of generality, suppose that both $\ket{\psi}^{AB}$ and $\ket{\phi}^{AB}$ are maximally correlated (i.e.~have the form $\ket{\psi}=\sum_i \sqrt{\psi_i}\ket{ii}$ and $\ket{\phi}=\sum_i\sqrt{\phi_i}\ket{ii}$).
Pad $\vec{\tau}(\psi)$ with enough zeros so that $\vec{\tau}(\psi)$ and $\vec{\tau}(\phi)$ are real vectors of equal length.  Since $\vec{\tau}(\psi)\prec\vec{\tau}(\phi)$, there exists a doubly stochastic matrix $D$ such that $\vec{\tau}(\psi)=D\vec{\tau}(\phi)$ \cite{Bhatia-2000a}.  Birkhoff's Theorem assures that $D=\sum_\alpha p_\alpha \Pi_\alpha$, where the $p_\alpha$ form a probability distribution and the $\Pi_\alpha$ are permutation matrices.  Then define the operators $M_\alpha:=\sqrt{p_\alpha}\Pi_\alpha^\dagger\bullet S$, where the elements of $S$ are given by $[[S]]_{ij}=\sqrt{\phi_i}/\sqrt{\psi_j}$ and ``$\bullet$'' denotes the Hadamard product.  Recall that the Hadamard product of two matrices $A$ and $B$ is the matrix $A\bullet B$ with elements $[[A\bullet B]]_{ij}=[[A]]_{ij}[[B]]_{ij}$.  Note that each $M_\alpha$ is an incoherent operator.  By construction $M_\alpha\otimes\Pi_\alpha \ket{\psi}\propto\ket{\phi}$ for every $\alpha$, and the relation $\vec{\tau}(\psi)=\sum_\alpha p_\alpha \Pi_\alpha\vec{\tau}(\phi)$ readily implies that $\sum_\alpha M_\alpha^\dagger M_\alpha=\mbb{I}$.  Hence, the protocol consists of Alice performing the incoherent measurement $\{M_\alpha\}_{\alpha}$, announcing her result, and then Bob performing the permutation $\Pi_\alpha$.
\end{proof}

\subsection{Proof of Theorem 4}

\begin{thm1}
The function $C_{\mc{L}}$ is an LIOCC monotone.
\end{thm1}
\begin{proof}
By the convex roof construction, it suffices to prove monotonicity for pure state transformations \cite{Vidal-2000a}.  To do so, we first introduce two relative entropy quantities for a general density matrix $\rho^S$ on system $S$ and a bipartite state $\rho^{AB}$ on joint system $AB$:   $C_r(\rho^S)=\min_{\sigma^S\in\mc{I}}S(\rho^S||\sigma^A)$ \cite{Baumgratz-2014a} and $C_r^{A|B}(\rho^{AB})=\min_{\sigma^{AB}\in\mc{QI}}S(\rho^{AB}||\sigma^{AB})$ \cite{Chitambar-2015a}, where $\mc{I}$ is the set of incoherent states for system $S$ and $\mc{QI}$ is the set of quantum-incoherent states for system $AB$.  For a pure state $\ket{\varphi}^{AB}$ with reduced density matrices $\varphi^A$ and $\varphi^B$, these quantities reduce to $C_r(\varphi^A)=S(A)_{\Delta(\varphi)}-\mathrm{E}(\varphi)$, $C_r(\varphi^B)=S(B)_{\Delta(\varphi)}-\mathrm{E}(\varphi)$, and $C_r^{A|B}(\varphi^{AB})=S(B)_{\Delta(\varphi)}$.  Furthermore, it was shown in Ref. \cite{Chitambar-2015a} that $C_r^{A|B}=C_a^{A|B}$ for pure states, where $C_a^{A|B}(\rho^{AB})$ the optimal asymptotic rate of coherence distillation on Bob's side when Alice helps.  Collecting these observations we therefore obtain
\begin{equation}
C_{\mc{L}}(\varphi^{AB})=C_r^{A|B}(\varphi^{AB})+C_r(\varphi^A)=C_r^{B|A}(\varphi^{AB})+C_r(\varphi^B).
\end{equation}
Now suppose that in the first round of the protocol, Alice makes a local measurement on the joint state $\ket{\Psi}^{AB}$ that generates an ensemble of pure state transformations $\ket{\Psi}^{AB}\to\{\ket{\omega_k}^{AB},p_k\}$.  Since $C_r^{A|B}$ is an LIOCC monotone, we have $C_r^{A|B}(\Psi^{AB})\geq \sum_kp_k C_r^{A|B}(\omega_k^{AB})$, and likewise because $C_r$ (for Alice's system) is a monotone under Alice's incoherent operation, we have $C_r(\Psi^A)\geq\sum_k p_k C_r(\omega_k^A)$.  Hence $C_{\mc{L}}(\Psi^{AB})\geq \sum_k p_k C_{\mc{L}}(\omega_k^{AB})$.  When Bob measures in the next round, we repeat the same argument on each $\omega_k^{AB}$ and use the fact that $C_{\mc{L}}(\omega_k^{AB})=C_r^{B|A}(\omega_k^{AB})+C_r(\omega_k^B)$.  By iteration, $C_{\mc{L}}$ behaves monotonically for all rounds of the protocol, and the theorem is proven.
\end{proof}

\subsection{Proof of Theorem 5}

\begin{thm1}
\label{thm5}
For a pure state $\ket{\Psi}^{AB}$ the following triples are achievable coherence-entanglement distillation rates
\begin{align}
(R_A,R_B,E^{co})&=\left(S(X)_{\Delta(\Psi)}-\mathrm{E}(\Psi),\;S(Y)_{\Delta(\Psi)},\;0\right)\label{Eq:DistTriple1}\\
(R_A,R_B,E^{co})&=\left(0,\;S(Y|X)_{\Delta(\Psi)},\;I(X:Y)_{\Delta(\Psi)}\right)\label{Eq:DistTriple4},
\end{align}
as well as the points obtained by interchanging $A\leftrightarrow B$ in Eqns. \eqref{Eq:DistTriple1} and \eqref{Eq:DistTriple4}.  Moreover, these points are optimal in the sense that any achievable rate triple must satisfy 
\begin{align}
R_A+R_B&\leq C_{\mc{L}}(\Psi)= S(X)_{\Delta(\Psi)}+S(Y)_{\Delta(\Psi)}-\mathrm{E}(\Psi), &R_B+E^{co}&\leq S(Y)_{\Delta(\Psi)}.\label{Eq:DistUPB}
\end{align}
\end{thm1}

\begin{proof}  The upper bound $R_A+R_B\leq C_{\mc{L}}(\Psi)$ follows from monotonicity of $C_{\mc{L}}$ under LIOCC and the fact that $C_{\mc{L}}$ is asymptotically continuous.  The latter property holds because $C_{\mc{L}}$ is defined on pure states in terms of relative entropy measures  and extended to mixed states using a convex roof (see \cite{Synack-2006a}).  The upper bound $R_B+E^{co}\leq S(Y)_{\Delta(\Psi)}$ follows from the first and again the fact that one eCoBit can be transformed into one local CoBit using LIOCC.

Moving to achievability, we first prove the rate triple given in Eq.~\eqref{Eq:DistTriple1}.  The protocol is based on decomposition \eqref{Eq:decomp2b}.  Starting from $\ket{\Psi}^{\otimes n}$ expressed in this form, Alice first measures the typical type encoded in register $A_1$ and (with high probability) announces the result $t$.  This leaves them with the state
\[\frac{1}{\sqrt{M_t(1-\epsilon)}}\sum_{m=1}^{M_t(1-\epsilon)}\ket{m}^{A_2}
\frac{1}{\sqrt{L_tC_t(1-3\epsilon)}}\sum_{l=1}^{L_t(1-3\epsilon)}\sum_{c=1}^{C_t}\left(V_{tm}\ket{lc}^{A_3A_4}\right)W_{tl}^\dagger \left(\Pi_{tlc}\ket{\psi_{tl_0c_0}}^{B}\ket{c}^{B_2}\right).\]
Alice next the performs the incoherent measurement $\{K_{lc}\}_{l,c=1}^{L_t(1-3\epsilon),C_t}$ on $A_3A_4$ where
\begin{equation}
K_{lc}=\op{l0}{l\gamma_c}V_{lm}^\dagger
\end{equation}
where $\ket{\gamma_c}=\sum_{c'=1}^{C_t}e^{i 2\pi \left(\tfrac{(c-1) (c'-1)}{C_t}\right)}\ket{c'}$.  She announces the result $(l,c)$, and then Bob performs the incoherent operation $W_{tl}^\dagger(\Pi_{tlc}^{B}\otimes\mbb{I}^{B_2})$ to systems $BB_2$ followed by the error-correction unitary $U_c=\sum_{c'=1}^{C_t}e^{-i 2\pi \left(\tfrac{(c-1) (c'-1)}{C_t}\right)}\ket{c'}$ performed on $B_2$.  The output state is
\begin{equation}
\frac{1}{\sqrt{M_t(1-\epsilon)}}\sum_{m=1}^{M_t(1-\epsilon)}\ket{m}^{A_2}\sum_{c=1}^{C_t}\ket{c}^{B_2}\ket{\psi_{tl_0c_0}}^{B}.
\end{equation}
Bob can further distill $\ket{\psi_{tl_0c_0}}\to\ket{+}^{rn}$ with $r\to S(Y|X)_{\Delta(\Psi)}$ as $n\to\infty$.  The total rates of coherence distillation thus approach Eq. \eqref{Eq:DistTriple1}.

We next turn to the achievability of rate triple Eq.~\eqref{Eq:DistTriple4}.  It is based on decomposition \eqref{Eq:decomp3c}.  Alice first does a type measurement and with high probability will generate the post-measurement state
\[\frac{1}{\sqrt{L_tC_t(1-3\epsilon)}}\sum_{l=1}^{L_t(1-3\epsilon)}\sum_{c=1}^{C_t}\ket{lc}^{A_2A_3} W_{tl}^\dagger \left(\Pi_{tlc}\ket{\psi_{tl_0c_0}}^{B}\ket{c}^{B_2}\right).\]
Alice then measures the code block $l$ on register $A_2$ and communicates the result to Bob.  He then performs the incoherent unitary $W_{tl}$ with the permutation $\Pi_{tlc}^{-1}$ conditioned on $\ket{c}^{B_2}$.  This generates the state
\begin{equation}
\frac{1}{\sqrt{C_t}}\sum_{c=1}^{C_t}\ket{c}^{A_3}\ket{c}^{B_2}\ket{\psi_{tl_0c_0}}^{B},
\end{equation}
which asymptotically approaches the desired rates of Eq.~\eqref{Eq:DistTriple4}.

\end{proof}

\begin{remark}
As noted in the main text, it is still unknown the optimal rate in which eCoBits can be distilled from a pure state using LIOCC.  Rate triple \eqref{Eq:DistTriple4} gives a rate of $I(X:Y)_{\Delta(\Psi)}$ with an additional coherence output rate of $S(Y|X)_{\Delta(\Psi)}$.  However, this point is not optimal in terms of the eCoBit rate.  The reason is that the quantity $I(X:Y)_{\Delta(\Psi)}$ can be increased by LIOCC.  As an example of this effect, consider the state
\[\ket{\Psi}^{AB}=\frac{1}{\sqrt{6}}\bigg(\ket{0}\otimes(\ket{0}+\ket{1}+\ket{2})+\ket{1}\otimes(\ket{0}-\ket{1}+\ket{2})\bigg).\]
It can be seen that $I(X:Y)_{\Delta(\Psi)}=0$.  However,
when Bob performs the incoherent measurement described by Kraus operators $\{K_0=\op{0}{+}+\op{1}{2},\;K_1=\op{0}{-}\}$, correlations are generated by measuring in the incoherent bases after Bob obtains outcome $K_0$.  Hence, this state has a nonzero eCoBit distillation rate.

We will now show that optimizing the mutual information $I(X:Y)_{\Delta(\Psi)}$ over all LIOCC protocols yields the maximum eCoBit distillation rate $E_{max}^{co}$.

\medskip

\noindent\textbf{Lemma.}
For a pure state $\ket{\Psi}^{AB}$, the optimal distillation rate of $\ket{\Phi_{AB}}$ is given by
\begin{equation}
E^{co}_D(\Psi)=\lim_{n\to\infty}\frac{1}{n}\sup_{\mc{L}}\sum_{m}p(m)I(X:Y)_{\Delta(\Psi_m)},
\end{equation}
where the supremum is taken over all LIOCC protocols that generate the multi-outcome transformation $\ket{\Psi}^{\otimes n} \to \{p(m),\ket{\Psi_m}\}_{m=1}^s$.

\begin{proof}
First let us prove sufficiency.  Consider any LIOCC protocol $\mc{L}$ that generates the pure state transformation $\Psi^{\otimes n}\to\{p(m),\Psi_m\}_{m=1}^s$.  Fix arbitrary $\epsilon,\delta>0$.  We consider $t$ blocks of $\Psi^{\otimes n}$ and perform $\mc{L}$ on each of the blocks. This is a standard technique used in quantum Shannon theory, and is often called ``double blocking''. For $t$ sufficiently large, with probability $>1-\epsilon$ the state obtained is $\bigotimes_{m=1}^s\Psi_m^{\otimes  tN_m}$ with $|N_m-p(m)|<\delta$.  This follows from the definition of $\delta$-typicality and Eq. \eqref{Eq:typicalProb}.  On each $\Psi_m^{\otimes tN_m}$ Alice and Bob perform the distillation protocol of Theorem \ref{thm5}, thus generating the state $\sigma_m$ where $||\sigma_m-\Phi_{AB}^{\otimes \lfloor tN_m (R_m-\epsilon)\rfloor}||<\epsilon$ and $R_m=I(X:Y)_{\Delta(\Psi_m)}$.  Hence in total we have the transformation $\Psi^{\otimes nt}\to\bigotimes_{m=1}^s \sigma_m$
where $||\bigotimes_{m=1}^s\sigma_m-\Phi_{AB}^{\otimes \sum_{m}\lfloor t N_m(R_m-\epsilon)\rfloor}||<s\epsilon$, from which we compute the rate
\begin{equation}
\frac{1}{nt} \sum_{m}\lfloor tN_m(R_m-\epsilon)\rfloor\geq\frac{1}{n}\sum_m p(m)I(X:Y)_{\Delta(\Psi}-O\left(\frac{\delta+\epsilon+s}{n}\right),
\end{equation}
where the additional terms come from $N_m> p(m)-\delta$ and the removal of $\lfloor \cdot\rfloor$.

We now turn to the converse.  Consider any LIOCC distillation protocol transforming $\Psi^{\otimes n}\to\sum_m p(m)\rho_m$ such that $F(\sum_m p(m)\rho_m,\Phi_{AB}^{\otimes nR})\geq 1-\epsilon$ (where $\rho_m^{AB}$ need not be pure).  Hence,
\begin{align}
\label{Eq:FidelityAvg}
(1-\epsilon)^2&\leq \sum_mp(m) F(\rho_m,\Phi^{nR}_{AB})^2\notag\\
&\leq \sum_m p(m) F(\Delta(\rho_m),\Delta(\Phi^{nR}_{AB}))^2.
\end{align}
Using Fannes' Inequality (Lemma \ref{lem:Fannes}), monotonicity of the trace norm under CPTP maps, and the relation $F(\rho,\sigma)^2\leq 1-\frac{1}{64}||\rho-\sigma||^4_1$, it is straightforward to show that
\begin{equation}
F(\Delta(\rho_m),\Delta(\Phi^{nR}_{AB}))^2\leq 1-\frac{1}{64}\left(\frac{|I(X:Y)_{\Delta(\rho_m)}-nR|-1}{3n\log d_Ad_B}\right)^4,\notag
\end{equation}
where we have also used the fact that $I(X:Y)_{\Delta(\Phi^{nR}_{AB})}=nR$.  Combining with Eq. \eqref{Eq:FidelityAvg} gives
\begin{align}
1-(1-\epsilon)^2&\geq\sum_mp(m)\frac{1}{64}\left(\frac{|nR-I(X:Y)_{\Delta(\rho_m)}|-1}{3n\log d_Ad_B}\right)^4\notag\\
&\geq \frac{1}{64}\left(\frac{|nR -\sum_mp(m)I(X:Y)_{\Delta(\rho_m)}|-1}{3n\log d_Ad_B}\right)^4.\notag
\end{align}
Therefore, we obtain
\begin{align}
\frac{1}{n}\sum_m p(m)I(X:Y)_{\Delta(\rho_m)}\geq R-3\log d_Ad_B(64[1-(1-\epsilon)])^{1/4}-1/n.\notag
\end{align}
This completes the proof. 
\end{proof}

\end{remark}

\subsection{Proof of Theorem 6}

\begin{thm1}
\begin{equation}
E^{LIOCC}_D(\Psi)=\mathrm{E}(\Psi).
\end{equation}
\end{thm1}

\begin{proof}
The protocol is based on decomposition (\ref{Eq:decomp4d}).  Quite simply, Alice measures the typical type $\ket{t}^{A_1}$ and codebook $\ket{\ovl{l}}^{A_2}$.  With high probability the post-measurement state will take the form
\begin{equation}
\frac{1}{\sqrt{\ovl{C_t}}}\sum_{\ovl{c}=1}^{\ovl{C_t}}\ket{\ovl{c}}^{A_3}\ovl{W}_{t\ovl{l}}^\dagger \left(\Pi_{t\ovl{l}\ovl{c}}\ket{\psi_{t\ovl{l_0}\ovl{c_0}}}^{B}\ket{\ovl{c}}^{B_2}\right).
\end{equation}
This is a maximally entangled state of approaching the desired size of $\ovl{C_t}\to \mathrm{E}(\Psi)$ as $n\to\infty$.
\end{proof}

\subsection{Proof of Theorem 7}

\begin{thm1}
\label{thm7}
A mixed state $\rho^{AB}$ has distillable entanglement if and only if entanglement can be distilled using LIOCC.
\end{thm1}
\begin{proof}
Note that an arbitrary quantum operation can be accomplished by unitary operations and incoherent projective measurements.  Thus, if $\mc{L}$ is a general LOCC operation such that $\mc{L}(\rho^{\otimes n})\approx \Phi_{AB}$, then, because of Lemma \ref{lem2}, there exists an LIOCC operation $\mc{L}_I$ consuming some finite amount of local coherence that transforms $\mc{L}_I(\rho^{\otimes n})=\mc{L}(\rho^{\otimes n})\approx \Phi_{AB}$.  Therefore, to asymptotically distill entanglement from $\rho$ by LIOCC, it suffices for Alice and Bob to first have a sufficient amount of local coherence.  Theorem 2 of Ref. \cite{Chitambar-2015a} implies that local coherence for either Alice or Bob can be distilled from $\rho^{AB}$ using LIOCC whenever $\rho^{AB}$ is entangled (see Remark below).  Hence, Alice and Bob first use $n_A$ copies of $\rho^{AB}$ to distill a sufficient amount of local coherence for Alice and an additional $n_B$ copies to distill sufficient coherence for Bob.  They can then implement $\mc{L}_I$ on $\rho^{\otimes n}$ with high precision, thus generating a close of approximation of $\Phi_{AB}$ using LIOCC operations and $n_A+n_B+n$ copies of $\rho$.
\end{proof}

\begin{remark}
Ref. \cite{Chitambar-2015a} deals with a more general setting in which the assisting party can perform arbitrary quantum operations.  However, the projective POVM described in Theorem 2 of \cite{Chitambar-2015a} can be implemented incoherently.  Indeed, if, say Alice, performs any projective measurement $\{\op{b_k}{b_k}\}_{k=0}^{d-1}$ with the $\ket{b_k}$ being orthonormal, Bob's post-measurement state will be the same if Alice were to instead perform the incoherent projective measurement $\{\op{k}{b_k}\}_{k=0}^{d-1}$.  See also Ref. \cite{Streltsov-2015b}.
\end{remark}

\end{document}